\documentclass[11pt]{article}
\newcommand{\LONG}[1]{#1}\newcommand{\SHORT}[1]{}

\usepackage{fullpage,times}
\usepackage{amssymb,amsthm}
\newtheorem{theorem}{Theorem}[section]
\newtheorem{lemma}[theorem]{Lemma}
\newtheorem{corollary}[theorem]{Corollary}
\newtheorem{rmk}[theorem]{Remark}
\newenvironment{remark}{\begin{rmk}\em}{\end{rmk}}

\newcommand{\eps}{\varepsilon}
\newcommand{\Z}{\mathbb{Z}}
\newcommand{\OO}{\widetilde{O}}
\newcommand{\CELL}{\textrm{cell}}
\newcommand{\IGNORE}[1]{}

 \usepackage{algpseudocode}
 \usepackage[]{algorithmicx}
\usepackage{xspace}
\usepackage[utf8]{inputenc}
\usepackage{graphicx}
\usepackage{amsmath}
\usepackage{float}
\usepackage{amsfonts}
\usepackage{algorithm}
\usepackage[breaklinks,bookmarks=false]{hyperref}
\usepackage{enumerate}

\newcommand{\polylog}{\rm{polylog}}

\begin{document}

\sloppypar
\title{Clustered Integer 3SUM via Additive Combinatorics}
\author{Timothy M. Chan%
\thanks{Cheriton School of Computer Science, University of Waterloo
(tmchan@uwaterloo.ca).  This research was done in part while
the author was on sabbatical in
Hong Kong University of Science and Technology.
The research is supported by an NSERC grant.}
 \and Moshe Lewenstein%
\thanks{Department of Computer Science, Bar-Ilan University
(moshe@cs.biu.ac.il).
This research was done in part while the author was on sabbatical in the University of Waterloo. The research is supported by BSF grant 2010437 and GIF grant 1147/2011.}}

\date{}
\maketitle
\setcounter{page}{0}
\thispagestyle{empty}

\begin{abstract}
We present a collection of new results on problems
related to 3SUM, including:
\begin{itemize}
\item The first truly subquadratic algorithm for
\begin{itemize}
\item computing the
(min,+) convolution for monotone increasing sequences with
integer values bounded by $O(n)$,
\item solving 3SUM for monotone sets in 2D with integer coordinates
bounded by $O(n)$, and
\item preprocessing a binary string for histogram indexing
(also called jumbled indexing).
\end{itemize}
The running time is
$O(n^{(9+\sqrt{177})/12}\,\textrm{polylog}\,n)=O(n^{1.859})$ with randomization, or $O(n^{1.864})$ deterministically.
This greatly improves the previous $n^2/2^{\Omega(\sqrt{\log n})}$
time bound obtained from Williams' recent result
on all-pairs shortest paths [STOC'14], and answers an open question raised by several researchers studying the histogram indexing problem.
\item The first algorithm for histogram indexing for any constant alphabet
size that achieves truly subquadratic preprocessing time and truly sublinear query time.
\item A truly subquadratic algorithm for integer 3SUM in the
case when the given set can be partitioned into $n^{1-\delta}$
clusters each covered by an interval of length $n$, for any
constant $\delta>0$.
\item An algorithm to preprocess any set of $n$ integers so that
subsequently 3SUM on any given subset can be solved in
$O(n^{13/7}\,\textrm{polylog}\,n)$ time.
\end{itemize}
All these results are
obtained by a surprising new technique,
based on the Balog--Szemer\'edi--Gowers Theorem from additive
combinatorics.
\IGNORE{
Additive combinatorics is the field concerned with combinatorial properties of sums (and differences) of sets. For example, rich structure can be shown for dense sets.

Several well known algorithmic problems are sum-set problems or neighboring problems. Perhaps, the most prominent is 3SUM, and its extension $k$-SUM. Another important neighboring problem is that of (min,+)-convolutions, which is strongly related to the All Pairs Shortest Path (APSP) problem.

We show how to apply results from additive combinatorics in order to obtain efficient algorithms for problems with sum-set nature. While there have been applications of additive combinatorics to computational complexity, these may be the first applications to algorithmic problems.

We first show how to efficiently implement an additive combinatoric result and then show truly subquadratic, that is $O(n^{2-\epsilon})$ time, algorithms for; (a) {\em 3SUM for monotone sets in $[n]^d$}, (b){\em monotone(min,+) convolution} (c) (the open problem of) a preprocessing algorithm for {\em histogram indexing}, (c) 3SUM when the data is clustered, (d) "3-SUM queries", left open in~\cite{BW12}.
}
\end{abstract}

\newpage

\section{Introduction}

\subsection{Motivation: Bounded Monotone (min,+) Convolution}

Our work touches on two of the most tantalizing open algorithmic
questions:
\begin{itemize}
\item Is there a truly subcubic ($O(n^{3-\delta})$-time) algorithm
for \emph{all-pairs shortest paths} (APSP) in general dense edge-weighted
graphs?  If all the edge weights are small integers bounded
by a constant, then the
answer is known to be yes, using fast matrix multiplication~\cite{ZwickSURVEY},
but the question remains open not only for arbitrary real weights,
but even for integer weights in, say, $[n]$.%
\footnote{$[n]$ denotes $\{0,1,\ldots,n-1\}$.}
The current best combinatorial algorithms run
in slightly subcubic $O((n^3/\log^2 n)(\log\log n)^{O(1)})$ time \cite{Chan10,HanTak}.  The recent breakthrough by
Williams~\cite{Williams14} achieves
$n^3 / 2^{\Omega(\sqrt{\log n})}$ expected time (using
fast rectangular matrix multiplication).
\item Is there a truly subquadratic ($O(n^{2-\delta})$-time) algorithm for the
\emph{3SUM} problem?  One way to state the problem (out of
several equivalent ways) is:
given sets $A,B,S$ of size $n$, decide whether
there exists a triple $(a,b,s) \in A\times B\times S$ such that $a+b=s$; in other words, decide whether
$(A+B)\cap S \not= \emptyset$.
All 3SUM algorithms we are aware of actually solve
a slight extension which we will call
the \emph{3SUM$^+$} problem:
decide for every element $s\in S$ whether
$a+b=s$ for some $(a,b)\in A\times B$; in other words,
report all elements in $(A+B)\cap S$.
If $A,B,S\subseteq [cn]$, then the problem can
be solved in $O(cn\log n)$ time by fast Fourier transform (FFT),
since it reduces to convolution for 0-1 sequences of length $cn$.
However, the question remains open
for general real values, or just integers from $[n]^2$.
The myriad ``3SUM-hardness'' results showing reductions
from both real and integer 3SUM to different problems
about computational geometry, graphs, and strings
\cite{GO95,Patrascu10,BDP08,PW10,VW09,AWW14,ACLL14,JV13,KPP14} tacitly
assume that the answer could be no.
The current best algorithm for integer 3SUM or 3SUM$^+$ by
Baran et al.~\cite{BDP08} runs in slightly subquadratic
$O((n^2/\log^2 n)(\log\log n)^2)$ expected time.
Gr\o nlund and Pettie's recent breakthrough for
general real 3SUM or 3SUM$^+$~\cite{GP14} achieves
$O((n^2/\log^{2/3}n)(\log\log n)^{2/3})$ deterministic
time or $O((n^2/\log n)(\log\log n)^2)$ expected time.
\end{itemize}

Our starting point concerns one of the most basic special cases---of both problems simultaneously---for which finding a truly subquadratic algorithm has remained open.  Put another way, solving the problem below is
a prerequisite towards solving APSP in truly subcubic
or 3SUM$^+$ in truly subquadratic time:

\begin{description}
\item[The Bounded Monotone (min,+) Convolution Problem:]
Given two monotone increasing sequences $a_0,\ldots,a_{n-1}$ and $b_0,\ldots,b_{n-1}$ lying in $[O(n)]$, compute their \emph{(min,+)
convolution} $s_0,\ldots,s_{2n-2}$, where
$s_k = \min_{i=0}^k (a_i+b_{k-i})$.
\end{description}

If all the $a_i$'s and $b_i$'s are small integers bounded by $c$, then (min,+) convolution can be reduced to classical convolution and can
thus be computed in $O(cn\log n)$ time by FFT\@.
If the differences $a_{i+1}-a_i$ and $b_{i+1}-b_i$ are randomly
chosen from $\{0,1\}$, then we can subtract a linear function $i/2$
from $a_i$ and $b_i$ to get sequences lying in a smaller
range $[\OO(\sqrt{n})]$
and thus solve the problem by FFT in $\OO(n^{3/2})$ time w.h.p.%
\footnote{
The $\OO$ notation hides polylogarithmic factors;
``w.h.p.'' means ``with high probability'', i.e., with probability
at least $1-1/n^c$ for input size $n$ and an arbitrarily large
constant $c$.
}
However, these observations do not seem to help in obtaining truly subquadratic worst-case time
for arbitrary bounded monotone sequences.

We reveal the connection to APSP and 3SUM$^+$:
\begin{itemize}
\item
A simple argument~\cite{BCDEHILPT14} shows that
(min,+) convolution can be reduced to
\emph{(min,+) matrix multiplication}, which in turn is known to
be equivalent to APSP\@.  More precisely, if we can
compute the (min,+) matrix multiplication of two $n\times n$ matrices, or solve APSP, in $T(n)$ time, then we can compute the (min,+) convolution of two sequences of length $n$ in $O(\sqrt{n}T(\sqrt{n}))$ time.  The APSP result by Williams immediately leads to an ${n^2 / 2^{\Omega(\sqrt{\log n})}}$-time algorithm for (min,+) convolution, the best result known
to date.  The challenge is to see if the bounded monotone case
can be solved more quickly.
\item
Alternatively, we observe that the bounded monotone (min,+) convolution problem can be reduced to 3SUM$^+$
for integer point sets in 2D, with at most a logarithmic-factor
slowdown, by setting
$A=\{(i,a): a_i\le a < a_{i+1}\}$ and
$B=\{(i,b): b_i\le b < b_{i+1}\}$ in $[O(n)]^2$, and
using $O(\log n)$ appropriately chosen sets $S$ via
a simultaneous binary search for all the minima (see Section~\ref{sec:mono:appl} for the
details).
Two-dimensional 3SUM$^+$ in $[O(n)]^2$ can be easily
reduced to one-dimensional 3SUM$^+$ in $[O(n^2)]$.
The current best result for integer 3SUM$^+$
leads to worse bounds, but the above reduction requires only
a special case of 3SUM$^+$, when the points in each of the 2D sets
$A,B,S$ in $[O(n)]^2$ form a monotone increasing sequence in both
coordinates simultaneously.  The hope is that the 3SUM$^+$
in this \emph{2D monotone} case can be solved more quickly.
\end{itemize}

The bounded monotone (min,+) convolution problem has a number
of different natural formulations and applications:
\begin{itemize}
\item
Computing the (min,+) convolution  for two integer sequences in
the \emph{bounded differences} case, where
$|a_{i+1}-a_i|,|b_{i+1}-b_i|\le c$ for some constant $c$,
can be reduced to the bounded monotone case by
just adding a linear function $ci$ to both $a_i$ and $b_i$.
(The two cases turn out to be equivalent\LONG{; see Remark~\ref{rmk-bounded}}.)
\item
Our original motivation concerns {\em histogram indexing}
(a.k.a.\ {\em jumbled indexing}) for a binary alphabet: the problem
is to preprocess a string $c_1\cdots c_n\in\{0,1\}^*$, so that we can decide whether there is a substring with exactly $i$ 0's and $j$ 1's for any given $i,j$ (or equivalently,
with length $k$ and exactly $j$ 1's for any given $j,k$).
Histogram indexing has been studied in over a dozen papers in the string algorithms literature in the last several years, and
the question of obtaining a truly subquadratic preprocessing
algorithm in the binary alphabet case has been raised several times
(e.g., see \cite{BCFL10,MR10,MR12} and the introduction of~\cite{ACLL14} for a more detailed survey).
In the binary case,
preprocessing amounts to computing the minimum number $s_k$
(and similarly the maximum number $s'_k$) of $1$'s
over all length-$k$ substrings for every $k$.
Setting $a_i$ to be the prefix sum $c_1+\cdots + c_i$, we see
that $s_k=\min_{i=k}^n (a_i-a_{i-k})$, which is precisely
a (min,+) convolution after negating and reversing the second
sequence.  The sequences are monotone increasing and lie in $\pm [n]$
(and incidentally also satisfy the bounded differences property).
Thus, binary histogram indexing can be reduced to bounded monotone
(min,+) convolution.
(In fact, the two problems turn out to be equivalent\LONG{; see Remark~\ref{rmk-jumble}}.)
\item
In another formulation of the problem, we are
given $n$ integers in $[O(n)]$ and want to find an interval
of length $\ell$ containing the smallest/largest number of elements,
for every $\ell\in [O(n)]$; or find an interval containing
$k$ elements with the smallest/largest length, for every $k\in [n]$.
This is easily seen to be equivalent to binary histogram indexing.
\item
For yet another application,
a ``necklace alignment'' problem
studied by Bremner et al.~\cite{BCDEHILPT14}, when restricted to input
sequences in $[O(n)]$, can also be reduced to bounded monotone (min,+) convolution.
\end{itemize}

\subsection{New Result}
We present the first truly subquadratic algorithm
for bounded monotone (min,+) convolution, and thus for
all its related applications such as binary histogram indexing.
The randomized version of our algorithm runs in
$\OO(n^{1.859})$ expected time; the curious-looking exponent is
more precisely $(9+\sqrt{177})/12$.
The deterministic version of the algorithm has a slightly worse
$O(n^{1.864})$ running time.  Our randomized algorithm uses FFT,
while
our deterministic algorithm uses both FFT and fast (rectangular) matrix multiplication.

\subsection{New Technique via Additive Combinatorics}

Even more interesting than the specific result
is our solution, which surprisingly relies on tools from a
different area: \emph{additive combinatorics}.  We explain
how we are led to that direction.

It is more convenient to consider the reformulation of
the bounded (min,+) monotone convolution problem, in terms
of solving 3SUM$^+$ over certain 2D monotone sets $A,B,S$
in $[O(n)]^2$, as mentioned earlier.
A natural approach to get a truly subquadratic algorithm
is via divide-and-conquer.  For example, we can
partition each input set into
subsets by considering a grid of side length $\ell$ and
taking all the nonempty grid cells;
because of monotonicity of the sets, there are $O(n/\ell)$ nonempty grid cells
each containing $O(\ell)$ points.
For every pair of a nonempty grid cell of $A$
and a nonempty grid cell of $B$, if their sum lands
in or near a nonempty grid cell of $S$, we need to
recursively solve the problem for the subsets of points in these
cells.  ``Usually'',
not many pairs out of the $O((n/\ell)^2)$ possible pairs would satisfy this condition and require recursive calls.
However, there are exceptions;
the most obvious case is when the nonempty grid cells of $A,B,S$ all lie on or near a line.
But in that case, we can subtract a linear function from
the $y$-coordinates to make all $y$-values small integers, and
then solve the problem by FFT directly!

Thus, we seek some combinatorial theorem roughly stating that
if many pairs of $A\times B$ have sum
in or near $S$, the sets $A$ and $B$ must be ``special'', namely,
close to a line.
It turns out that
the celebrated {\em Balog--Szemer\'edi--Gowers Theorem} (henceforth,
the BSG Theorem) from additive combinatorics accomplishes exactly
what we need.
One version of
the BSG Theorem (out of several different versions) states:
\begin{quote}
Given sets $A,B,S$ of size
$N$ in any abelian group such that $|\{(a,b)\in A\times B: a+b\in S\}|\ge \alpha N^2$,
we must have $|A'+B'| \le O((1/\alpha)^5 N)$ for some
large subsets $A'\subseteq A$ and $B'\subseteq B$ with
$|A'|,|B'|\ge \Omega(\alpha N)$.
\end{quote}
(We will apply the theorem to the sets of nonempty grid cells
in $\Z^2$, with $N=O(n/\ell)$.)

The original proof by
Balog and Szemer\'edi~\cite{BS94}
used heavy machinery, namely, the regularity lemma, and
had a much weaker superexponential dependency on $\alpha$.
A much simpler proof with a polynomial $\alpha$-dependency later appeared in a (small part of a famous)
paper by Gowers~\cite{Gowers01}.  Balog~\cite{Balog07} and Sudakov
et al.~\cite{SSV94} further refined the factor to the stated $(1/\alpha)^5$.
Since then, the theorem has appeared in books~\cite{TV06} and surveys~\cite{Lovett14,Viola11}.
Although additive combinatorics, and specifically the BSG Theorem, have found some applications in theoretical computer science
before~\cite{Lovett14,Viola11} (for example, in
property testing~\cite{BLR}), we are not aware of any applications in classical algorithms---we believe this adds further interest to our work.

Four points about the BSG Theorem statement are relevant
to our algorithmic applications:
\begin{itemize}
\item
First, as it reveals, the right criterion of ``special'' is not that
the two sets $A$ and $B$ are close to a line, but rather that
their sumset $A+B$ has small size.
According to another celebrated
theorem from additive combinatorics, {\em Freiman's Theorem}~\cite{Fre,TV06},
if a sumset $A+A$ has size $O(|A|)$, then $A$
indeed has special structure in the sense that it must
be close to a projection of a higher-dimensional
lattice.  Fortunately, we do not need this theorem (which
requires a more complicated proof and has
superexponential $\alpha$-dependency): if $A+B$ has small size,
we can actually compute $A+B$ by FFT directly, as explained in the
``FFT Lemma'' of Section~\ref{sec:bsg}.
\item
Second, the theorem does not state that $A$ and $B$ themselves
must be special,
but rather that we can extract large subsets $A'$ and $B'$ which are special.
In our applications, we need to ``cover'' all possible pairs
in $\{(a,b)\in A\times B:a+b\in S\}$, and so we need a stronger version of the
BSG Theorem which allows us to remove already covered pairs and iterate.  Fortunately,
Balog~\cite{Balog07} and Szemer\'edi et al.~\cite{SSV94} provided
a version of the BSG Theorem that did precisely this;
see Section~\ref{sec:bsg} for the precise statement.
The resulting corollary on pairs covering is dubbed
the ``BSG Corollary'' in Section~\ref{sec:bsg}.
\item
The BSG Theorem was originally conceived with the setting
of constant $\alpha$ in mind, but
polynomial $\alpha$-dependency (which fortunately we have)
will be critical in obtaining truly
subquadratic algorithms in our applications, as we need to
choose $\alpha$ (and $\ell$) to balance the contribution of
the ``usual'' vs.\ ``special'' cases.
\item
The BSG Theorem is originally
a mathematical result, but the time complexity of the
construction will matter in our applications.
We present, to our knowledge, the first time bounds in
Theorem~\ref{runtime-corollary}.
\end{itemize}

Once all the components involving the BSG Corollary and the FFT Lemma
are in place, our main algorithm for bounded monotone (min,+)
convolution can be described simply, as revealed in Section~\ref{sec:mono}.

\subsection{Other Consequences of the New Technique}

The problem we have started with, bounded monotone (min,+) convolution, is just one of many applications that can
be solved with this technique.  We briefly list our
other results:
\begin{itemize}
\item
We can solve 3SUM$^+$ not only in the 2D monotone case, but
also in the $d$-dimensional monotone case
in truly subquadratic
$\OO(n^{(11-d+\sqrt{(d-11)^2+48d})/12})$ expected time
for any constant $d$ (Theorem~\ref{thm-monotone}).
If only $A$ and $B$ are monotone,
a slightly weaker bound $\OO(n^{2-2/(d+13)})$ still holds;
if just $A$ is monotone, another
weaker bound $\OO(n^{2-1/(d+6)})$ holds (Theorem~\ref{cor-monotone-offline}).
\item
In 1D,
we can solve integer 3SUM$^+$ in truly subquadratic
$n^{2-\Omega(\delta)}$ time
if the input sets are {\em clustered} in the sense that
they can be covered by $n^{1-\delta}$
intervals of length $n$ (Corollary~\ref{cor-general}).
In fact, just one of the sets $A$ needs to be clustered.
This is the most general setting of 3SUM we know that
can be solved in truly subquadratic time
(hence, the title of the paper).
In some sense, it ``explains'' all the other results.
For example, $d$-dimensional
monotone sets, when mapped down to 1D in an appropriate way, become
clustered integer sets.

\item We can also solve a data structure version of 3SUM$^+$ where
$S$ is given online:
preprocess $A$ and $B$ so that we can decide whether
any query point $s$ is in $A+B$.
For example, if $A$ and $B$ are monotone in $[n]^d$,
we get truly subquadratic $\OO(n^{2-\delta})$ expected
preprocessing time and truly sublinear $\OO(n^{2/3+\delta(d+13)/6})$
query time for any sufficiently small $\delta>0$
(\LONG{Corollary~\ref{cor-monotone-online}}\SHORT{see the full paper}).
\item
As an immediate application,
we can solve the histogram indexing problem for any constant alphabet size $d$: we can preprocess any string $c_1\cdots c_n\in [d]^*$ in truly
subquadratic $\OO(n^{2-\delta})$ expected time, so that we can decide
whether there is a substring whose vector of character counts
matches exactly a query vector in truly sublinear
$\OO(n^{2/3+\delta(d+13)/6})$ time for any sufficiently
small $\delta>0$ (\LONG{Corollary~\ref{cor-jumble-online}}\SHORT{see the full paper}).
This answers an open question and improves a
recent work by Kociumaka et al.~\cite{KRR13}.
Furthermore, if $n$ queries are given offline, we can answer all queries in total $\OO(n^{2-2/(d+13)})$ expected time
(Corollary~\ref{cor-jumble-offline}).
As $d$ gets large,
this upper bound approaches a conditional lower bound recently
shown by Amir et al.~\cite{ACLL14}.
\item
For another intriguing consequence, we can preprocess
any universes $A_0,B_0,S_0\subseteq\Z$ of $n$ integers so that given any subsets
$A\subseteq A_0,B\subseteq B_0, S\subseteq S_0$, we can solve 3SUM$^+$ for $A,B,S$ in truly subquadratic
$\OO(n^{13/7})$ time (\LONG{Theorem~\ref{thm-preproc1}}\SHORT{see the full paper}).  Remarkably, this is a result about
general integer sets.  One of the results in
Bansal and Williams' paper~\cite{BW12} mentioned precisely
this problem but obtained much weaker polylogarithmic speedups.
When $S_0$ is not given, we can still achieve $\OO(n^{1.9})$ time
(\LONG{Theorem~\ref{thm-preproc2}}\SHORT{see the full paper}).
\end{itemize}

\section{Ingredients: The BSG Theorem/Corollary
and FFT Lemma}\label{sec:bsg}


As noted in the introduction,
the key ingredient behind all of our results is the Balog--Szemer\'edi--Gowers Theorem.  Below,
we state the particular version of the theorem we need, which can be found
in the papers by Balog~\cite{Balog07} and Sudakov et al.~\cite{SSV94}.  A complete proof is redescribed in
\LONG{Sections \ref{sec-graph} and \ref{sec-bsg-proof}}\SHORT{the full paper}.

\begin{theorem} {\bf(BSG Theorem)}\ \ Let $A$ and $B$ be finite subsets of an abelian group, and $G\subseteq A \times B$.
  Suppose that $|A||B|=\Theta(N^2)$, $|\{a+b: (a,b)\in G\}| \le tN$, and $|G| \ge \alpha N^2$.
  Then there exist subsets $A'\subseteq A$ and $B'\subseteq B$ such that

\begin{itemize}
\item[\rm (i)]   $|A' + B'| \:\leq\:  O((1/\alpha)^5 t^3 N)$, and
\item[\rm (ii)]  $|G\cap (A'\times B')| \:\geq\:
\Omega(\alpha |A'||B|) \:\geq\: \Omega(\alpha^2 N^2)$.
\end{itemize}
\end{theorem}

The main case to keep in mind is when $|A|=|B|=N$ and $t=1$,
which is sufficient for many of our applications, although the
more general ``asymmetric'' setting does arise in at least two of the
applications.

In some versions of the BSG Theorem,
$A=B$ (or $A=-B$) and we further insist that $A'=B'$ (or $A'=-B'$); there, the $\alpha$-dependencies are a bit worse.

In some simpler versions of the BSG Theorem that appeared in many papers (including the version mentioned in the introduction), we are not given $G$,
but rather a set $S$ of size $tN$ with $|\{(a,b):a+b\in S\}|\ge
\alpha N^2$; in other words, we are considering the special case $G=\{(a,b):a+b\in S\}$.  Condition (ii) is replaced by
$|A'|,|B'|\ge\Omega(\alpha N)$.
For our applications, it is crucial to consider the version with
a general $G$.  This is because of the need to apply the
theorem iteratively.

If we apply the theorem iteratively, starting with
$G = \{(a,b): a+b\in S\}$ for a given set $S$, and repeatedly removing
$A' \times B'$ from $G$, we obtain the following corollary, which
is the combinatorial result we will actually use in all our applications
(and which, to our knowledge, has not been stated explicitly before):

\begin{corollary}~\label{cor-BSG} {\bf(BSG Corollary)}\ \
Let $A,B,S$ be finite subsets of an abelian group.
  Suppose that $|A||B|=O(N^2)$ and $|S| \le tN$.  For any $\alpha<1$, there exist subsets
  $A_1,\ldots,A_k\subseteq A$ and $B_1,\ldots,B_k\subseteq B$ such that

\begin{enumerate}
\item[\rm (i)]
the \emph{remainder set} $R = \{(a,b)\in A\times B: a+b\in S\}
\setminus \bigcup_{i=1}^k (A_i \times B_i)$ has
size at most $\alpha N^2$,
\item[\rm (ii)]
  $|A_i + B_i| \:\le\: O((1/\alpha)^5 t^3 N)$ for each $i=1,\ldots,k$, and
\item[\rm (iii)] $k = O(1/\alpha)$.
\end{enumerate}
\end{corollary}

A naive argument gives only $k=O((1/\alpha)^2)$, as each
iteration removes $\Omega(\alpha^2 |A||B|)$ edges from $G$,
but a slightly more refined analysis, given
in \LONG{Section~\ref{sec-bsg-corollary-proof}}\SHORT{the full paper}, lowers
the bound to $k=O(1/\alpha)$.

None of the previous papers on the BSG Theorem addresses the running
time of the construction, which will of course be
important for our
algorithmic applications.  A polynomial time bound can be easily
seen from most known proofs of the BSG Theorem, and
is already sufficient to yield some nontrivial result
for bounded monotone (min,+)-convolution and
binary histogram indexing.  However, a quadratic time bound is necessary to get
nontrivial results for other applications such as histogram
indexing for larger alphabet sizes.
In \LONG{Sections \ref{sec-graph-time} and \ref{sec-bsg-corollary-time}}\SHORT{the full paper, using a number of additional
ideas (e.g., sampling tricks for sublinear algorithms)}, we
show that the construction in the BSG Theorem/Corollary
can indeed be done in near quadratic time
with randomization, or in matrix multiplication time deterministically.


\newcommand{\MM}{{\cal M}}

\begin{theorem}~\label{runtime-corollary}
In the BSG Corollary, the subsets $A_1,\ldots,A_k,B_1,\ldots,B_k$, the remainder set $R$,
and all the sumsets $A_i+B_i$ can be constructed by
\begin{itemize}
\item[\rm (i)] a deterministic algorithm in time $O((1/\alpha)^{0.4651}N^{2.3729})$,
or more precisely, $O((1/\alpha)\MM(\alpha|A|,|A|,|B|))$,
where $\MM(n_1,n_2,n_3)$ is the complexity of multiplying
an $n_1\times n_2$ and an $n_2\times n_3$ matrix, or
\item[\rm (ii)] a randomized Las Vegas algorithm in expected
time $\OO(N^2)$ for $t\ge 1$, or $\OO(N^2 + (1/\alpha)^5|A|)$  otherwise.
\end{itemize}
\end{theorem}


We need one more ingredient.  The BSG Theorem/Corollary produces
subsets that have small sumsets.  The following lemma shows
that if the sumset is small, we can compute
the sumset efficiently:

\begin{lemma} {\bf(FFT Lemma)}\ \ Given sets $A,B\subseteq[U]^d$
of size $O(N)$ for a constant $d$ with $|A+B|\le O(N)$, and
given a set $T$ of size $O(N)$
which is known to be a superset of $A+B$, we can compute $A+B$ by
\begin{itemize}
\item[\rm (i)] a randomized Las Vegas algorithm in $\OO(N)$ expected time, or
\item[\rm (ii)] a deterministic algorithm that runs in $\OO(N)$
time after preprocessing $T$ in $\OO(N^{1+\eps})$ time for
an arbitrarily small constant $\eps>0$.
\end{itemize}
\end{lemma}

As the name indicates, the proof of the lemma uses fast Fourier transform. The randomized version
was proved by Cole and Hariharan~\cite{CH02}, who
actually obtained a more general result where the superset $T$
need not be given: they addressed the problem
of computing the (classical) convolution of two sparse vectors and
presented a Las Vegas algorithm that runs in time sensitive to the
number of nonzero entries in the output vector; computing
the sumset $A+B$ can be viewed as an instance of the
sparse convolution problem and can be solved by their algorithm
in $O(|A+B|\log^2 N)$ expected time.
Amir et al.~\cite{AKP07} have given a derandomization technique
for a related problem (sparse wilcard matching), which can also
produce a deterministic algorithm for computing $A+B$
in the setting when $T$ is given and has been preprocessed, but the preprocessing of $T$ requires $\OO(N^2)$ time.

In \LONG{Section~\ref{sec-fft}}\SHORT{the full paper},
we give self-contained proofs of both the randomized
and deterministic versions of the FFT Lemma.  For the randomized
version, we do not need the extra complications of Cole and
Hariharan's algorithm, since $T$ is given in our applications.
For the deterministic version, we significantly reduce
Amir et al.'s preprocessing cost to $\OO(N^{1+\eps})$, which is
of independent interest.

\IGNORE{
\noindent
{\bf Proof outline:}

First we can replace $Z\times Z$ with $Z$, by mapping
  $(a_1,a_2)$ to $Ma_1+a_2$ for a sufficiently large $M$.

  Use a small number of "nearly additive" hash functions
  $h_j:Z \rightarrow [O(N \polylog N)]$
  (e.g., $h_j(a) = a \mod p_j$ for some random prime $p_j$, or
  the hash function from the Baran-Demaine-Patrascu 3SUM paper). NEEDS MORE WORK.

  For each $j$, multiply the polynomials
  $\sum_{a \in A} x^{h_j(a)}$ and $\sum_{b \in B} x^{h_j(b)}$ with a sparse convolution.

  A sparse convolution is a convolution whose input has a high number of zero values. Hence, the input is represented as a list of the non-zero values, and its size $O(n)$ = the number of non-zero values. The size of the (non-zero) output is denoted by $k$. Cole and Hariharan~\cite{CH02} claimed a randomized algorithm that runs in time $\tilde{O}(k)$ yielding the desired.

  In Section~\ref{determinstic-algo} we show our claim for the deterministic case. We point out that a weaker deterministic result that runs in $O(n^2)$ time~\cite{AKP07} existed previously. AMIR-PORAT DO NOT REALLY NEED T, IS IT COMPARABLE?
  \qed

}

\section{3SUM$^+$ for Monotone Sets in $[n]^d$}\label{sec:mono}

We say that a set in $\Z^d$ is \emph{monotone (increasing/decreasing)}
if it can be written as $\{a_1,\ldots,a_n\}$ where
the $j$-th coordinates of
$a_1,\ldots,a_n$ form a monotone (increasing/decreasing) sequence
for each $j=1,\ldots,d$.
Note that a monotone set in $[n]^d$ can have size at most $dn$.

\subsection{The Main Algorithm}

\begin{theorem}\label{thm-monotone}
Given monotone sets $A,B,S\subseteq [n]^d$ for a constant~$d$,
we can solve 3SUM$^+$ by
\begin{itemize}
\item[\rm (i)] a randomized Las Vegas algorithm in expected time
$\OO(n^{(9+\sqrt{177})/12})=O(n^{1.859})$ for $d=2$,
$\OO(n^{(8+\sqrt{208})/12})=O(n^{1.869})$ for $d=3$,
or more generally,
$\OO(n^{(11-d+\sqrt{(d-11)^2+48d})/12})$ for any~$d$, or
\item[\rm (ii)] a deterministic algorithm in time
$O(n^{1.864})$ for $d=2$,
$O(n^{1.901})$ for $d=3$,
$O(n^{1.930})$ for $d=4$,
$O(n^{1.955})$ for $d=5$,
$O(n^{1.976})$ for $d=6$,
or $O(n^{1.995})$ for $d=7$.
\end{itemize}
\end{theorem}
\begin{proof}
Divide $[n]^d$ into $O((n/\ell)^d)$
grid cells of side length $\ell$, for some parameter $\ell$ to be set later.
Define $\CELL(p)$ to be a label (in $\Z^d$) of the grid cell
containing the point $p$; more precisely,
$\CELL(x_1,\ldots,x_d) := (\lfloor x_1/\ell\rfloor,\ldots, \lfloor x_d/\ell\rfloor)$.

We assume that all points $(x_1,\ldots,x_d)\in A$
satisfy $x_j\bmod{\ell} < \ell/2$ for every $j=1,\ldots, d$;
when this is true, we say
that $A$ is \emph{aligned}.  This is without loss of generality,
since $A$ can be decomposed into a constant ($2^d$) number of subsets, each of which is a translated copy of an aligned set,
by shifting selected coordinate positions by $\ell/2$.
Similarly, we may assume that $B$ is aligned.
By alignedness, the following property holds: for any $a\in A$ and
$b\in B$, $s=a+b$ implies $\CELL(s)=\CELL(a)+\CELL(b)$.

Our algorithm works as follows:
\begin{description}
\item[Step 0:]
Apply the BSG Corollary to the sets $A^*=\{\CELL(a):a\in A\}$,
$B^*=\{\CELL(b):b\in B\}$, $S^*=\{\CELL(s):s\in S\}$.
This produces subsets $A_1^*,\ldots,A_k^*,B_1^*,\ldots,B_k^*$ and a remainder set $R^*$.

Note that $|A^*|,|B^*|,|S^*| =O(n/\ell)$ by monotonicity of $A,B,S$.
The parameters in the BSG Corollary are thus $N=\Theta(n/\ell)$
and $t=1$.
Hence, this step takes $\OO((n/\ell)^2)$ expected time
by Theorem~\ref{runtime-corollary}.
\item[Step 1:]
For each $(a^*,b^*)\in R^*$,
recursively solve the problem for the sets
$\{a\in A: \CELL(a)=a^*\}$, $\{b\in B: \CELL(b)=b^*\}$,
$\{s\in S: \CELL(s)= a^*+b^*\}$.

Note that
this step creates $|R^*|=O(\alpha(n/\ell)^2)$ recursive calls,
where each set lies in a smaller universe, namely,
a translated copy of $[\ell]^d$.
\item[Step 2:]
For each $i=1,\ldots,k$,
apply the FFT Lemma to generate
$\{a\in A: \CELL(a)\in A_i^*\} + \{b\in B: \CELL(b)\in B_i^*\}$,
which is contained in the
superset $T_i = \{s\in \Z^d: \CELL(s)\in A_i^*+B_i^*\}$.
Report those generated elements that are in $S$.

Note that the size of $A_i^*+B_i^*$ is $O((1/\alpha)^5 n/\ell)$,
and so the size of $T_i$ is $O((1/\alpha)^5 n/\ell\cdot \ell^d)$.
As $k=O(1/\alpha)$, this step takes
$\OO((1/\alpha)^6 n\ell^{d-1})$ expected time.
\end{description}

Correctness is immediate from the BSG Corollary, since
$\{(a^*,b^*)\in A^*\times B^*: a^*+b^*\in S^*\}$ is covered by
$R^*\cup \bigcup_{i=1}^k (A_i^*\times B_i^*)$.

The expected running time is characterized by the following interesting recurrence:
\[ T(n) \:\le\: \OO((n/\ell)^2) \:+\: O(\alpha (n/\ell)^2)\, T(\ell) \:+\:
                 \OO((1/\alpha)^6 n\ell^{d-1}).
\]
Note that the reduction to the aligned case increases only
the hidden constant factors in the three terms.
We can see that this recurrence leads to truly subquadratic
running time for any constant $d$---even if we use the trivial upper bound $T(\ell)=O(\ell^2)$
(i.e., don't use recursion)---by setting $\ell$ and $1/\alpha$ to be some sufficiently small powers of $n$.

For example, for $d=2$, we can set $\ell=n^{0.0707}$ and $1/\alpha=n^{0.1313}$ and obtain
\[ T(n) \:\le\: O(n^{1.8586}) \:+\: O(n^{1.7273})\, T(n^{0.0707}),
\]
which solves to $O(n^{1.859})$.

More precisely, the recurrence solves to $T(n)=\OO(n^z)$
by setting $\ell=n^x$ and $1/\alpha=n^y$ for $x,y,z$ satisfying
the system of equations
$z=2(1-x)=-y+2(1-x)+xz=6y+(1-x)+dx$.  One can check that the
solution for $z$ indeed obeys the quadratic equation
$6z^2+(d-11)z-2d=0$.

\IGNORE{
\begin{theorem}
Alternatively, we can solve the problem in Theorem~\ref{thm-monotone} in
$O(n^{z+o(1)})$ deterministic time where
$z$ is the larger root of
$(6-\mu)z^2 + (d(1+\mu) -13 + \omega + 2\mu)z - d(\omega + 2\mu)=0$.
Here, $\mu=(3-\rho-\omega)/(1-\rho)$, and
$\omega$ and $\rho$ are
the square and rectangular matrix multiplication exponents.
\end{theorem}
\begin{proof}
}

Alternatively, the deterministic version of the algorithm
has running time
given by the recurrence
\[ T(n) \:\le\: O((1/\alpha)^\mu (n/\ell)^\nu) \:+\: O(\alpha (n/\ell)^2)\, T(\ell) \:+\:
                 O((1/\alpha)^6 n^{1+\eps}\ell^{d-1}),
\]
with $\mu=0.4651$ and $\nu=2.3729$,
which can be solved in a similar way.
The quadratic equation now becomes
$(6-\mu)z^2 + ((1+\mu)d -13 + \nu + 2\mu)z - (\nu + 2\mu)d=0$.
\end{proof}

\IGNORE{
For example, for $d=2$, the randomized time bound is
$\OO(n^{(9+\sqrt{177})/12})=O(n^{1.859})$ (attained
by setting $\ell \approx n^{0.0707}$ and $1/\alpha\approx n^{0.1313}$)
and the
deterministic time bound is $O(n^{1.864})$ (attained
by setting $\ell \approx n^{0.2353}$ and $1/\alpha \approx n^{0.1046}$).
}
As $d$ gets
large, the exponent in the randomized bound is
$2-2/(d+13) - \Theta(1/d^3)$; however, the deterministic bound
is subquadratic only for $d\le 7$ using
the current matrix multiplication exponents
(if $\omega=2$, then we would have subquadratic deterministic
time for all $d$).
In the remaining applications, we will mostly emphasize randomized bounds for the sake of simplicity.

\IGNORE{

w := 2.3728639;
r := 0.30298;
mu := (3-r-w)/(1-r);
d := 2;
solve({z = mu*y+w*t, z = -y+2*t+(1-t)*z, z = 6*y + t + d*(1-t)});

w := 2.3728639;
r := 0.30298;
mu := (3-r-w)/(1-r);
z := proc(d)  a := 6-mu; b := d*(1+mu)-13+w+2*mu; c := -d*(w+2*mu);
              (-b + sqrt(b^2 - 4*a*c))/(2*a); end;
> z(2);
                           1.863158551

> z(3);
                           1.900083074

> z(4);
                           1.929989472

> z(5);
                           1.954841478

> z(6);
                           1.975898027

> z(7);
                           1.994013663

> z(8);
                           2.009793911

}

\subsection{Application to the 2D Connected Monotone Case,
Bounded Monotone (min,+) Convolution, and
Binary Histogram Indexing}\label{sec:mono:appl}

\SHORT{
As noted in the full paper,
bounded monotone (min,+) convolution reduces to
3SUM$^+$ for 2D monotone sets, and is equivalent to
(min,+) convolution in the bounded differences case and
to binary histogram indexing.  Thus, all these
problems can now be solved in $O(n^{1.859})$ expected time (or
$O(n^{1.864})$ deterministic time).
}
\LONG{
We say that a set in $\Z^d$ is \emph {connected}
if every two points in the set are connected by a path
using only vertices from the set and edges of unit $L_1$-length.
In the case of connected monotone sets $A,B$ in 2D, we show how
to compute a complete representation of $A+B$.

\begin{corollary}\label{cor-monotone-2d}
Given connected monotone increasing sets $A,B\subseteq [n]^2$,
we can compute the boundary of $A+B$, a region bounded by
two monotone increasing sets,
in $O(n^{1.859})$ expected time (or
$O(n^{1.864})$ deterministic time).
\end{corollary}
\begin{proof}
First we show that $A+B$ is indeed a region bounded by
two monotone increasing sets.
Define $I_k$ to be the set of $y$-values of $A+B$ at
the vertical line $x=k$.
Then each $I_k$ is a single interval: to see this,
express $I_k$ as the union of intervals
$I_k^{(i)}=\{y: (i,y)\in A\}+\{y: (k-i,y)\in B\}$ over all $i$, and just observe that
each interval $I_k^{(i)}$ overlaps with the
next interval $I_k^{(i+1)}$ as $A$ and $B$ are connected and monotone increasing.  Since the lower/upper endpoints of $I_k$ are clearly monotone increasing in $k$, the conclusion follows.

We reduce the problem to 3SUM$^+$ for three 2D monotone sets.
We focus on the lower boundary of $A+B$, i.e.,
computing the lower endpoint of the interval $I_k$, denoted
by $s_k$, for all $k$.   The upper
boundary can be computed in a symmetric way.
We compute all $s_k$ by a simultaneous binary search in
$O(\log n)$ rounds as follows.

In round $i$, divide $[2n]$ into grid intervals of length $2^{\lceil \log (2n)\rceil - i}$.  Suppose that at the beginning
of the round, we know which grid interval $J_k$ contains $s_k$
for each $k$.
Let $m_k$ be the midpoint of~$J_k$.
Form the set $S=\{(k,m_k) : k\in [2n]\}$.
Since the $s_k$'s are monotone increasing, we know that
the $J_k$'s and $m_k$'s are as well; hence, $S$ is a
monotone increasing set in $[2n]^2$.
Apply the 3SUM$^+$ algorithm to $A,B,S$.
If $(k,m_k)$ is found to be in $A+B$, then
$I_k$ contains $m_k$ and so we know that $s_k\le m_k$.
Otherwise, $I_k$ is either completely smaller than $m_k$ or
completely larger than $m_k$; we can tell which is the case
by just comparing any one element of $I_k$ with $m_k$, and so we know whether $s_k$ is smaller or
larger than $m_k$.  (We can easily
pick out one element from $I_k$
by picking any $i$ in the $x$-range of $A$ and $j$ in the $x$-range of $B$ with $i+j=k$, picking any point of $A$ at $x=i$ and any
point of $B$ at $x=j$, and summing their $y$-coordinates.)
We can now reset $J_k$ to the
half of the interval that we know contains $m_k$, and proceed
to the next round.  The total running time is that of the
3SUM$^+$ algorithm multiplied by $O(\log n)$.
\end{proof}

It is possible to modify the algorithm in Theorem~\ref{thm-monotone}
directly to prove the corollary and avoid the extra
logarithmic penalty, but the preceding black-box reduction
is nevertheless worth noting.

\begin{corollary}
Given two monotone increasing sequences
$a_0\ldots,a_{n-1}\in [O(n)]$ and $b_0,\ldots,b_{n-1}\in [O(n)]$,
we can compute their (min,+) convolution in
$O(n^{1.859})$ expected time (or
$O(n^{1.864})$ deterministic time).
\end{corollary}
\begin{proof}
We just apply Corollary~\ref{cor-monotone-2d} to
the connected monotone increasing sets $A=\{(i,a): a_i\le a < a_{i+1}\}$
and $B=\{(i,b): b_i\le b < b_{i+1}\}$ in $[O(n)]^2$.  Then the lowest $y$-value
in $A+B$ at $x=k$ gives the $k$-th entry of the (min,+) convolution.
\end{proof}

\begin{remark}\label{rmk-convol}
The problems in the preceding two corollaries are
in fact equivalent.
To reduce in the other direction, given connected monotone increasing sets
$A,B\subseteq [n]^2$, first we may assume that the $x$-ranges
of $A$ and $B$ have the same length, since
we can prepend one of the sets with a horizontal line segment
without affecting the lower boundary of $A+B$.  By translation,
we may assume that the $x$-ranges of $A$ and $B$ are identical
and start with 0.  We define the monotone increasing
sequences $a_i'=$ (lowest $y$-value of
$A$ at $x=i$) and $b_i'=$ (lowest $y$-value of $B$ at $x=i$).
Then the (min,+) convolution of the two sequences gives
the lower boundary of $A+B$.  The upper boundary can be computed
in a symmetric way.
\end{remark}

\begin{remark}\label{rmk-bounded}
We can now compute the (min,+) convolution of two integer
sequences with bounded differences property,
by reducing to the monotone case as noted in the introduction.

This version is also equivalent.
To reduce in the other direction, given connected monotone increasing sets
$A=\{a_1,\ldots,a_{|A|}\}$ and $B=\{b_1,\ldots,b_{|B|}\}$ in $[n]^2$
where $a_{i+1}-a_i,b_{i+1}-b_i\in\{(1,0),(0,1)\}$,
we apply the linear
transformation $\phi(x,y) = (x+y,y)$.
After the transformation,
$\phi(a_{i+1})-\phi(a_i),\phi(b_{i+1})-\phi(b_i)\in
\{(1,0),(1,1)\}$.  When applying the same
reduction in Remark~\ref{rmk-convol}
to the transformed sets $\phi(A)$ and $\phi(B)$,
the two resulting monotone increasing sequences will satisfy
the bounded differences property (the
differences of consecutive elements are all in $\{0,1\}$).
The boundary of $A+B$ can be inferred from
the boundary of $\phi(A)+\phi(B)$.
\end{remark}

\begin{corollary}\label{cor-jumble-binary}
Given a string $c_1\cdots c_n\in\{0,1\}^*$,
we can preprocess in
$O(n^{1.859})$ expected time (or
$O(n^{1.864})$ deterministic time) into an $O(n)$-space structure,
so that we can answer
histogram queries, i.e., decide whether there exists
a substring with exactly $i$ $0$'s and $j$ $1$'s for any
query values $i,j$, in $O(1)$ time.
\end{corollary}
\begin{proof}
One reduction to bounded monotone (min,+) convolution has been noted
briefly in the introduction.  Alternatively, we can just
apply Corollary~\ref{cor-monotone-2d} to the
connected monotone increasing sets $A=\{a_0,\ldots,a_n\}\subseteq [n]^2$ and $B=-A$,
where the $x$- and $y$-coordinates of $a_i$ are the number of
0's and 1's in the prefix $c_1\cdots c_i$. (We can make $B$
lie in $[n]^2$ by translation.)  Then $a_j-a_i\in A+B$ gives
the number of 0's and 1's in the substring $c_i\cdots c_{j-1}$ for any $i < j$.  The boundary of $A+B$ gives the desired structure.
\end{proof}

\begin{remark}\label{rmk-jumble}
This problem is also equivalent.
To reduce in the other direction, suppose we have
connected monotone increasing sets $A=\{a_1,\ldots,a_{|A|}\}$ and
$B=\{b_1,\ldots,b_{|B|}\}$ in $[n]^2$, given in sorted order with $a_1=b_1=(0,0)$.
We set $c_i=0$ if $a_{i+1}-a_i=(1,0)$,
or $1$ if $a_{i+1}-a_i=(0,1)$; and set $d_i=0$ if
$b_{i+1}-b_i=(1,0)$, or $1$ if $b_{i+1}-b_i=(0,1)$.
We then solve histogram indexing for
the binary string $c_{|A|-1}\cdots c_1 0^n d_1\cdots d_{|B|-1}$.
The minimum number of 1's over all substrings with $n+k$ 0's
(which can be found in $O(\log n)$ queries by binary search)
gives us the lowest point in $A+B$ at $x=k$.  The upper boundary
can be computed in a symmetric way.
\end{remark}

}

\section{Generalization to Clustered Sets}\label{sec:clustered}

In the main algorithm of the previous section, monotonicity is
convenient but not essential.
In this section, we identify the sole property needed: clusterability.  Formally, we say that
a set in $\Z^d$ is \emph{$(K,L)$-clustered} if it can be
covered by $K$ disjoint hypercubes each of volume $L$.
We say that it is \emph{$(K,L,M)$-clustered} if
furthermore each such hypercube contains at most $M$ points of the set.

\subsection{The Main Algorithm}

\begin{theorem}\label{thm-cluster}
Given $(K_A,L,M_A)$-, $(K_B,L,M_B)$-, and $(K_S,L,M_S)$-clustered
sets $A,B,S\subseteq\Z^d$ for a constant $d$,
we can solve 3SUM$^+$ in expected time
$$ \OO(K_A K_B + (K_AK_B)^{5/7} K_S^{3/7} L^{1/7}  W^{6/7}
          + K_A (K_B W)^{5/6}),$$
where $W = \min\{M_AM_B, M_AM_S, M_BM_S\}$.
\end{theorem}
\begin{proof}
The algorithm is similar to the one in Theorem~\ref{thm-monotone},
except without recursion.
We use a grid of side length $\ell := \lceil L^{1/d}\rceil$, and
as before,
we may assume that $A$ and $B$ are aligned.
\begin{description}
\item[Step 0:]
Apply the BSG Corollary to the sets $A^*=\{\CELL(a):a\in A\}$,
$B^*=\{\CELL(b):b\in B\}$, $S^*=\{\CELL(s):s\in S\}$.
This produces subsets $A_1^*,\ldots,A_k^*,B_1^*,\ldots,B_k^*$ and a remainder set $R^*$.

Note that $|A^*|=O(K_A),|B^*|=O(K_B),|S^*| =O(K_S)$.
The parameters in the BSG Corollary
are thus $N=\Theta(\sqrt{K_AK_B})$ and $t=O(K_S/\sqrt{K_AK_B})$.
This step takes $\OO(K_AK_B+(1/\alpha)^5 K_A)$ expected time
by Theorem~\ref{runtime-corollary}.
\item[Step 1:]
For each $(a^*,b^*)\in R^*$,
solve the problem for the sets
$\{a\in A: \CELL(a)=a^*\}$, $\{b\in B: \CELL(b)=b^*\}$,
$\{s\in S: \CELL(s)= a^*+b^*\}$.

Note that the three sets have sizes $O(M_A)$, $O(M_B)$,
$O(M_S)$, and so the naive brute-force algorithm which
tries all pairs from two of the three sets takes $\OO(W)$ time.
As $|R^*|=O(\alpha N^2)=O(\alpha K_AK_B)$, this step takes total
time $O(\alpha K_AK_BW)$.

\item[Step 2:]
For each $i=1,\ldots,k$,
apply the FFT Lemma to generate
$\{a\in A: \CELL(a)\in A_i^*\} + \{b\in B: \CELL(b)\in B_i^*\}$,
which is contained in the
superset $T_i = \{s\in \Z^d: \CELL(s)\in A_i^*+B_i^*\}$.
Report those generated elements that are in $S$.

Note that the size of $A_i^*+B_i^*$ is $O((1/\alpha)^5 t^3 N)$,
and so the size of $T_i$ is $O((1/\alpha)^5 t^3 NL)$.
As $k=O(1/\alpha)$, this step takes expected time
$\OO((1/\alpha)^6 t^3 NL)
=\OO((1/\alpha)^6 K_S^3 L/(K_AK_B))$.
\end{description}

The total expected time is
$\OO(K_AK_B + \alpha K_AK_B W + (1/\alpha)^6 K_S^3 L/(K_AK_B)
+ (1/\alpha)^5 K_A)$.
We set $1/\alpha = \min\{[(K_AK_B)^2 W /(K_S^3 L)]^{1/7},\,
(K_B W)^{1/6}\}$.
\end{proof}

It turns out that the $M$ bounds on the number of points
per cluster are not essential, and neither is clusterability
of the third set $S$.  In fact, clusterability of only one
set $A$ is enough to obtain nontrivial results.

\begin{corollary}\label{cor-cluster}
Given sets $A,B,S\subseteq\Z^d$ of size $O(n)$ for a constant $d$ where
$A$ and $B$ are $(K_A,L)$- and $(K_B,L)$-clustered,
we can solve 3SUM$^+$ in expected time

$$\OO(K_AK_B + n^{12/7}(K_AL)^{1/7}).$$
\end{corollary}
\begin{proof}
We say that a set of size $n$ is \emph{equitably $(K,L)$-clustered} if
it is $(K,L,O(n/K))$-clustered.  Suppose that $A,B,S$
are equitably  $(K_A,L)$-, $(K_B,L)$-, and $(K_S,L)$-clustered.
Then in Theorem~\ref{thm-cluster}, we set $M_A=O(n/K_A)$,
$M_B=O(n/K_B)$, $M_S=O(n/K_S)$, and
upper-bound $W$ in the second term by the following
weighted geometric mean
(with carefully chosen weights):
\[ ((n/K_A)(n/K_B))^{1/2} ((n/K_A)(n/K_S))^{1/6} ((n/K_B)(n/K_S))^{1/3} \:=\: n^2/(K_A^{2/3}K_B^{5/6}K_S^{1/2});
\]
and we upper-bound $W$ in the third term more simply by $(n/K_A)(n/K_B)$.
This leads to the expression
$\OO(K_AK_B + n^{12/7}(K_A L)^{1/7} + n^{5/3}K_A^{1/6})$.
The third term is always dominated by the second (since $K_A\le n$),
and so we get precisely the stated time bound.

What if $A,B,S$ are not equitably clustered?
We can decompose $A$ into $O(\log n)$
equitably $(\le K_A,L)$-clustered subsets: just put points in
hypercubes with between $2^i$ and $2^{i+1}$ points into the $i$-th
subset.  We can do the same for $B$ and $S$.
The total time increases by at most an $O(\log^3 n)$ factor
(since the above bound is nondecreasing in $K_A$ and $K_B$
and independent of $K_S$).
\end{proof}

The corollary below now follows immediately by
just substituting $K_A=n^{1-\delta}$, $L=n$, and $K_B\le n$.

\begin{corollary}\label{cor-general}
Given sets $A,B,S\subseteq\Z^d$ of size $O(n)$ for a constant $d$ where
$A$ is $(n^{1-\delta},n)$-clustered,
we can solve 3SUM$^+$ in $n^{2-\Omega(\delta)}$ expected time.
\end{corollary}

Although it may not give the best quantitive bounds, it
describes the most general setting under which
we know how to solve 3SUM in truly subquadratic time.

For example, for $d=1$, the above corollary generalizes the well-known
fact that 3SUM for integer sets in $[n^{2-\delta}]$ can be
solved in subquadratic time (by just doing one FFT), and
a not-so-well-known fact that 3SUM for three integer sets where
only one set is assumed to be in $[n^{2-\delta}]$ can still be solved
in subquadratic time (by doing several FFTs, without requiring additive combinatorics---a simple exercise we leave to the reader).

\SHORT{
Although we have stated the above results in $d$ dimensions,
the one-dimensional case of integers contains the essence, since we can
map higher-dimensional clustered sets to 1D, as noted in the full
paper.
}
\LONG{
\begin{remark}
Although we have stated the above results in $d$ dimensions,
the one-dimensional case of integers contains the essence, since we can
map higher-dimensional clustered sets to 1D\@.
Specifically,
consider a grid of side length $\ell := \lceil L^{1/d}\rceil$, and
without loss of generality,
assume that $A,B\subseteq [U]^d$ are aligned.
We can map each point $(x_1,\ldots,x_d)$ to the integer
\[
 L\cdot \sum_{j=1}^d \lfloor x_j/\ell\rfloor\cdot (2\lceil U/\ell\rceil)^{j-1}
 + \sum_{j=1}^d (x_j\bmod{\ell})\cdot \ell^{j-1}.
\]
If $A$ is $(K,L)$-clustered, then the mapped set in 1D is
still $(O(K),L)$-clustered.  By alignedness,
3SUM$^+$ solutions are preserved by the mapping.
\end{remark}
}

\subsection{Application to the Monotone Case and Offline Histogram
Queries}

\begin{corollary}\label{cor-monotone-offline}
Given sets $A,B,S\subseteq [n]^d$ of size $O(n)$ for a constant~$d$
where $A$ and $B$ are monotone,
we can solve 3SUM$^+$ in
$\OO(n^{2-2/(d+13)})$ expected time.

If only $A$ is monotone, we can solve the
problem in $\OO(n^{2-1/(d+6)})$ expected time.
\end{corollary}
\begin{proof}
A monotone set in $[n]^d$ is $(O(n/\ell),\ell^d)$-clustered
for all $\ell$.  For example, it is $(O(n^{1-1/d},n)$-clustered,
and so by Corollary~\ref{cor-general}, we know that truly
subquadratic running time is achievable.
For the best quantitive bound, we set
$K_A,K_B=O(n/\ell)$ and $L=\ell^d$ in Corollary~\ref{cor-cluster}
and get
$\OO(n^2/\ell^2 + n^{12/7}(n\ell^{d-1})^{1/7}).$
We set $\ell = n^{1/(d+13)}$ to balance the two terms.

If only $A$ is monotone, we get
$\OO(n^2/\ell + n^{12/7}(n\ell^{d-1})^{1/7}).$
We set $\ell = n^{1/(d+6)}$.
\end{proof}

The above bounds are (predictably) slightly worse than in
Theorem~\ref{thm-monotone}, which assumes
the monotonicity of the third set $S$.  The algorithm there also
exploits a stronger ``hierarchical'' clustering property enjoyed by
monotone sets (namely, that clusters are themselves clusterable),
which allows for recursion.


\begin{corollary}\label{cor-jumble-offline}
Given a string $c_1\cdots c_n\in [d]^*$ for a constant
alphabet size $d$ and a set $S\subseteq [n]^d$ of $O(n)$ vectors,
we can answer histogram queries, i.e.,
decide whether there exists a substring with
exactly $i_0$ $0$'s, \ldots, and $i_{d-1}$ $(d-1)$'s,
for all the vectors $(i_0,\ldots,i_{d-1})\in S$,
in $\OO(n^{2-2/(d+13)})$ total expected time.
\end{corollary}
\begin{proof}
We just apply the 3SUM$^+$ algorithm in Corollary~\ref{cor-monotone-offline} to the
\LONG{connected }monotone increasing sets $A=\{a_0,\ldots,a_n\}\subseteq [n]^d$ and
$B=-A$, and the (not necessarily monotone) set $S$,
where the $d$ coordinates of $a_i$ hold the number of 0's, \ldots, $(d-1)$'s
in the prefix $c_1\cdots c_i$.
Then the $d$ coordinates of $a_j-a_i\in A+B$ give
the number of 0's, \ldots, $(d-1)$'s in the substring $a_i\cdots a_{j-1}$ for any $i < j$.
\end{proof}

The above upper bound nicely complements
the conditional hardness results by
Amir et al.~\cite{ACLL14}.  They proved an $n^{2-4/(d-O(1))}$ lower
bound on the histogram problem under the assumption that integer 3SUM$^+$ requires at least $n^{2-o(1)}$ time, and an $n^{2-2/(d-O(1))}$
lower bound under a stronger assumption that
3SUM$^+$ in $[n]^2$ requires at least $n^{2-o(1)}$ time.  (Their results
were stated for online queries but hold in the offline
setting.)  On the other hand, if the assumption fails, i.e.,
integer 3SUM$^+$ turns out to have a truly subquadratic algorithm,
then there would be a truly subquadratic algorithm
for offline histogram queries with an exponent independent of $d$.

\LONG{

\section{Online Queries}\label{sec:online}

We now show how the same techniques can even
be applied to the setting where the points of the third set $S$ are not
given in advance but arrive online.

\subsection{The Main Algorithm}

\begin{theorem}\label{thm-cluster-online}
Given $(K_A,L,M_A)$- and $(K_B,L,M_B)$-clustered sets $A,B\subseteq\Z^d$ for a constant $d$ and a parameter $P$,
we can preprocess in expected time
$$\OO(K_AK_B +  (K_AK_B)^{8/7}(M_AM_B)^{6/7}L^{1/7}/P^{3/7}
  + K_A(K_BM_AM_B)^{5/6})$$
into a data structure with $O(K_AK_B + K_AK_BL/P)$ space,
so that we can decide whether any query point $s$ is in $A+B$
in $\OO(\min\{M_A,M_B\}\cdot P)$ time.
\end{theorem}
\begin{proof}
The approach is similar to our previous algorithms, but with
one more idea: dividing into the cases of ``low popularity''
and ``high popularity'' cells.
As before, we use a grid of side length $\ell := \lceil L^{1/d} \rceil$ and assume
that $A$ and $B$ are aligned.

The preprocessing algorithm works as follows:
\begin{description}
\item[Step 0:] Let $A^*=\{\CELL(a):a\in A\}$ and
$B^*=\{\CELL(b):b\in B\}$.
Place each $(a^*,b^*)\in A^*\times B^*$
in the \emph{bucket} for $s^*=a^*+b^*$.  Store all these
buckets.  Define the \emph{popularity}
of $s^*$ to be the number of elements in its
bucket.  Let $S^*$ be the set of all $s^*\in \Z^d$
with popularity $>P$.
Apply the BSG Corollary to $A^*,B^*,S^*$.

Note that $|A^*|=O(K_A)$, $|B^*|=O(K_B)$, and $|S^*|=O(K_AK_B/P)$, because
the total popularity over all possible $s^*$ is
at most $O(K_AK_B)$.  The parameters in the BSG Corollary
are thus $N=\Theta(\sqrt{K_AK_B})$ and $t=|S^*|/N = O(\sqrt{K_AK_B}/P)$.
The buckets can be formed in $\OO(K_AK_B)$ time and space.
This step takes $\OO(K_AK_B + (1/\alpha)^5 K_A)$ expected time
by Theorem~\ref{runtime-corollary}.
\item[Step 1:]
For each $(a^*,b^*)\in R^*$, generate the list
$\{a\in A: \CELL(a)=a^*\} + \{b\in B: \CELL(b)=b^*\}$.

Note that naively each such list can
be generated in $O(M_AM_B)$ time.
Since $|R^*|=O(\alpha N^2)=O(\alpha K_AK_B)$, this step takes total
time $O(\alpha K_AK_BM_AM_B)$.

\item[Step 2:]
For each $i=1\ldots,k$, apply the FFT Lemma to generate the list
$\{a\in A: \CELL(a)\in A_i^*\} + \{b\in B: \CELL(b)\in B_i^*\}$,
which is contained in the
superset $T_i = \{s\in \Z^d: \CELL(s)\in A_i^*+B_i^*\}$.

Note that the size of $A_i^*+B_i^*$ is $O((1/\alpha)^5 t^3 N)$,
and so the size of $T_i$ is $O((1/\alpha)^5 t^3 N L)$.
As $k=O(1/\alpha)$, this step takes expected time
$\OO((1/\alpha)^6 t^3 NL)
= \OO((1/\alpha)^6 (\sqrt{K_AK_B}/P)^3 \sqrt{K_AK_B} L)
= \OO((1/\alpha)^6 (K_AK_B)^2L/P^3)$.
\item[Step 3:]
Store the union ${\cal L}$ of all the lists generated in Steps
1 and~2.
Prune elements not in $\{s\in\Z^d: \CELL(s)\in S^*\}$ from~${\cal L}$.

Note that the pruned list ${\cal L}$ has size at most
$|S^*| L = O(K_AK_B L/P)$.
\end{description}

The query algorithm for a given point $s$ works as follows:
\begin{description}
\item[``Low'' Case:] $\CELL(s)$ has popularity $\le P$.
W.l.o.g., assume $M_A\le M_B$.
We look up the bucket for $\CELL(s)$.  For each $(a^*,b^*)$
in the bucket, we search for some $a\in A$ with $\CELL(a)=a^*$
that satisfies $s-a\in B$.  The search time is $\OO(M_A)$ per bucket entry, for a total of $\OO(M_AP)$.
\item[``High'' Case:] $\CELL(s)$ has popularity $>P$.  We just test $s$ for membership in ${\cal L}$ in
$\OO(1)$ time.
\end{description}

To summarize, the preprocessing time is
$\OO(K_AK_B + \alpha K_AK_BM_AM_B + (1/\alpha)^6 (K_AK_B)^2L/P^3
 + (1/\alpha)^5 K_A)$,
the space usage is $O(K_AK_B + K_AK_BL/P)$,
and the query time is $\OO(M_AP)$.
We set $1/\alpha = \min\{[(M_AM_BP^3)/(K_AK_B L)]^{1/7},\,
  (K_BM_AM_B)^{1/6}\}$.
\end{proof}

\begin{corollary}\label{cor-cluster-online}
Given $(K_A,L)$- and $(K_B,L)$-clustered sets $A,B\subseteq\Z^d$ of size $O(n)$ for a constant $d$ and a parameter $Q$,
we can preprocess in expected time
$$\OO(K_AK_B + n^{12/7} (K_AL)^{1/7}Q^{3/7})$$
into a data structure with $\OO(K_AK_B + K_ALQ)$ space,
so that we can decide whether any query element $s$ is in $A+B$
in $\OO(n/Q)$ time.
\end{corollary}
\begin{proof}
Recall the definition of equitable clustering in the proof of Corollary~\ref{cor-cluster}.
Suppose that $A,B,S$
are equitably  $(K_A,L)$-, $(K_B,L)$-, and $(K_S,L)$-clustered.
Then in Theorem~\ref{thm-cluster-online}, setting $M_A=O(n/K_A)$,
$M_B=O(n/K_B)$, $M_S=O(n/K_S)$, and the parameter
$P=\max\{K_A,K_B\}/Q\ge K_A^{1/3}K_B^{2/3}/Q$,
we get the desired preprocessing time
$\OO(K_AK_B + n^{12/7} (K_AL)^{1/7}Q^{3/7} + n^{5/3}K_A^{1/6})$
(the last term is always dominated by the second),
space $O(K_AK_B + K_ALQ)$,
and query time $\OO(n/Q)$.

We can reduce to the equitable case as in the proof of Corollary~\ref{cor-cluster}, by decomposing each set into
$O(\log n)$ subsets.
\end{proof}

\subsection{Application to the Monotone Case and
Online Histogram Queries}

\begin{corollary}\label{cor-monotone-online}
Given two monotone sets $A,B\subseteq [n]^d$ for a constant~$d$
and a parameter $\delta$,
we can preprocess in $\OO(n^{2-\delta})$ expected time,
so that we can decide whether any query point $s$ is in $A+B$
in $\OO(n^{2/3 + \delta (d+13)/6})$ time.

If only $A$ is monotone, the query time is
$\OO(n^{2/3+\delta (d+6)/3})$.
\end{corollary}
\begin{proof}
A monotone set in $[n]^d$ is $(O(n/\ell),\ell^d)$-clustered
for all $\ell$.  We set $K_A,K_B=O(n/\ell)$ and $L=\ell^d$
in Corollary~\ref{cor-cluster-online}
and get $\OO(n^2/\ell^2 + n^{12/7}(n\ell^{d-1})^{1/7}Q^{3/7})$
preprocessing time and $\OO(n/Q)$ query time.
We set $\ell=n^{\delta/2}$ and $Q=n^{1/3-\delta (d+13)/6}$.

If only $A$ is monotone, we get
 $\OO(n^2/\ell + n^{12/7}(n\ell^{d-1})^{1/7}Q^{3/7})$
preprocessing time and $\OO(n/Q)$ query time.
We set $\ell=n^{\delta}$ and $Q=n^{1/3-\delta (d+6)/3}$.
\end{proof}

If we want to balance the preprocessing cost with the
cost of answering $n$ queries, in the case when $A$
and $B$ are both monotone, we can set
$\delta=2/(d+19)$ and obtain $\OO(n^{2-2/(d+19)})$ preprocessing
time and $\OO(n^{1-2/(d+19)})$ query time.
These bounds are (predictably) slightly worse than in Corollary~\ref{cor-monotone-offline} in the offline setting.

\begin{corollary}\label{cor-jumble-online}
Given a string $c_1\cdots c_n\in [d]^*$ for a constant
alphabet size $d$,
we can preprocess in $\OO(n^{2-\delta})$ expected time, so that
we can answer histogram queries, i.e.,
decide whether there exists a substring with
exactly $i_0$ $0$'s, \ldots, and $i_{d-1}$ $(d-1)$'s,
for any query vector $(i_0,\ldots,i_{d-1})$,
in $\OO(n^{2/3 + \delta (d+13)/6})$ time.
\end{corollary}
\begin{proof}
We just apply Corollary~\ref{cor-monotone-online} to the same sets
$A$ and $B$ from the proof of Corollary~\ref{cor-jumble-offline}.
\end{proof}

\begin{remark}
The idea of dividing into the cases of low and high popularity cells has
previously been used by Kociumaka et al.~\cite{KRR13} for histogram indexing,
but they were able to obtain only a space/query-time tradeoff,
namely, a data structure with $\OO(n^{2-\delta})$ space
and $\OO(n^{\delta (2d-1)})$ query time.
Their data structure requires close to quadratic preprocessing time.
Incidentally, we can easily improve their space/query-time tradeoff:
substituting $K_A,K_B=O(n/\ell)$ and $L=\ell^d$
in Corollary~\ref{cor-cluster-online} gives
$\OO(n^2/\ell^2 + n\ell^{d-1}Q)$ space and $\OO(n/Q)$ time.
Setting $\ell=n^{\delta/2}$ and $Q=n^{1-\delta(d+1)/2}$
then gives $\OO(n^{2-\delta})$ space and $\OO(n^{\delta(d+1)/2})$
query time.  All this does not require additive combinatorics
(which helps only in improving the preprocessing time).
\end{remark}

\section{3SUM$^+$ in Preprocessed Universes}\label{sec:preprocessed}

As one more application, we show that 3SUM can be solved in
truly subquadratic time if the universe has been preprocessed.
(Note, though, that the running time below is subquadratic in the
size of the three given universes $A_0,B_0,S_0$, and not
of $A,B,S$.)  This version of the problem was
considered by Bansal and Williams~\cite{BW12}, who only obtained time
bounds of the form $n^2/\textrm{polylog}\,n$.

\begin{theorem}\label{thm-preproc1}
Given sets $A_0,B_0,S_0\subseteq\Z$ of size $O(n)$, we can
preprocess in $\OO(n^2)$ expected time into a data structure
with $O(n^{13/7})$ space, so that given any sets $A\subseteq A_0$,
$B\subseteq B_0$, $S\subseteq S_0$, we can solve 3SUM$^+$
for $A,B,S$ in $\OO(n^{13/7})$ time.
\end{theorem}
\begin{proof}
Our algorithm works as follows:
\begin{description}
\item[Preprocessing:] Apply the BSG
Corollary to $A_0,B_0,S_0$.  Store the resulting subsets $A_1,\ldots,A_k,$ $B_1,\ldots,B_k$ and remainder set $R$, and
also store each $T_i=A_i+B_i$.

The expected preprocessing time is $\OO(n^2)$ by Theorem~\ref{runtime-corollary}.
As $|R|\le \alpha n^2$, $|T_i|=O((1/\alpha)^5 n)$, and $k=O(1/\alpha)$,
the total space usage is $O(\alpha n^2 + (1/\alpha)^6 n)$.
\end{description}

Now, given $A\subseteq A_0$, $B\subseteq B_0$, $S\subseteq S_0$,
we do the following:
\begin{description}
\item[Step 1:] For each $(a,b)\in R$, if $a\in A$, $b\in B$, and
$a+b\in S$, then report $a+b$.
This step takes $O(|R|)=O(\alpha n^2)$ time.

\item[Step 2:] For each $i=1,\ldots,k$, apply the FFT Lemma to generate
$(A_i\cap A) + (B_i\cap B)$, which is contained in the superset $T_i$.  Report those generated elements that are in $S$.
This step takes $\OO((1/\alpha)^6 n)$ expected time.
\end{description}

The total expected time
is $\OO(\alpha n^2 + (1/\alpha)^6 n)$.  We set $1/\alpha = n^{1/7}$
to balance the two terms.
The part after preprocessing can be made deterministic,
after including
an extra $O((1/\alpha)^6 n^{1+\eps})=o(n^2)$ cost for preprocessing
the $T_i$'s for the deterministic version of the FFT Lemma.
\end{proof}

In the above theorem, $S$ is superfluous, since
we may as well take $S=S_0$ when solving 3SUM$^+$.  In the
next theorem, we show that a slightly weaker
subquadratic time holds if the universe for $S$ is not given in advance.

\begin{theorem}\label{thm-preproc2}
Given sets $A_0,B_0\subseteq\Z$ of size $O(n)$, we can
preprocess in $\OO(n^2)$ expected time into a data structure
with $O(n^2)$ space, so that given any sets $A\subseteq A_0$,
$B\subseteq B_0$, $S\subseteq \Z$ of size $O(n)$, we can solve 3SUM$^+$
for $A,B,S$ in $\OO(n^{19/10})$ time
\end{theorem}
\begin{proof}
We incorporate one more idea: dividing into the cases of
low and high popularity.

\begin{description}
\item[Preprocessing:]
Place each $(a,b)\in A_0\times B_0$ in the \emph{bucket} for
$s=a+b$.  Store all these buckets.
Define the \emph{popularity}
of $s$ to be the size of its bucket.
Let $S_0$ be the set of all
elements $s\in\Z$ with popularity $>n/t$.  Apply the BSG
Corollary to $A_0,B_0,S_0$, store the resulting subsets $A_1,\ldots,A_k,B_1,\ldots,B_k$ and remainder set $R$, and
also store each $T_i=A_i+B_i$.

Note that $|S_0|=O(tn)$,
because the total popularity is $O(n^2)$.
The buckets can be formed in $O(n^2)$ time and space.
The expected preprocessing time is $\OO(n^2)$
by Theorem~\ref{runtime-corollary}.
As $|R|\le \alpha n^2$, $|T_i|=O((1/\alpha)^5 t^3n)$, and $k=O(1/\alpha)$,
the total space usage is $O(n^2 + (1/\alpha)^6 t^3 n)$.
\end{description}

Now, given $A\subseteq A_0$, $B\subseteq B_0$, $S\in\Z$ of size $O(n)$, we do the following:
\begin{description}
\item[Step 0:] For each $s\in S$ of popularity $\le n/t$, we look
up the bucket for $s$, and report $s$ if some $(a,b)$ in the bucket
has $a\in A$ and $b\in B$.  The search time is $O(n/t)$ per
element in $S$, for a total of $O(n^2/t)$.

\item[Step 1:] For each $(a,b)\in R$, if $a\in A$, $b\in B$, and
$a+b\in S$, then report $a+b$.
This step takes $O(|R|)=O(\alpha n^2)$ time.

\item[Step 2:] For each $i=1,\ldots,k$, apply the FFT Lemma to generate
$(A_i\cap A) + (B_i\cap B)$, which is contained in the superset $T_i$.  Report those generated elements that are in $S$.
This step takes $\OO((1/\alpha)^6 t^3 n)$ expected time.
\end{description}

The total expected time
is $\OO(n^2/t + \alpha n^2 + (1/\alpha)^6 t^3n)$.  We set
$t=1/\alpha = n^{1/10}$
to balance the three terms.  Again the part after preprocessing
can be made deterministic.
\end{proof}

\begin{remark}
The above theorem holds for real numbers as well, if we assume an unconventional model of computation for the preprocessing algorithm.
To reduce the real case
to the integer case, first sort
$A_0+B_0$ and compute
the smallest distance $\delta$ between any two elements in $A_0+B_0$ in $\OO(n^2)$ time.  Divide the real line into grid intervals
of length $\delta/4$.  Without loss of generality, assume that
$A_0$ and $B_0$ are aligned.   Replace each real number $x$
in $A_0$ and $B_0$ with the integer $f(x)=\lfloor x/(\delta/4)\rfloor$.  Then for any $a\in A_0$, $b\in B_0$, and $s\in A_0+B_0$, $s=a+b$ iff $f(s)=f(a)+f(b)$.  This reduction however requires the floor function and
working with potentially very large integers afterwards.
\end{remark}

\begin{corollary}
Given a vertex-weighted graph $G=(V,E)$ with $n$ vertices, we can decide
whether there exists a (not necessarily induced) copy of
$K_{1,3}$ (the star with four nodes) that has total
weight exactly equal to a given value $W$ in $\OO(n^{2.9})$ time.
\end{corollary}
\begin{proof}
Let $w(v)$ denote the weight of $v$.
Preprocess $A_0=B_0=\{w(v): v\in V\}$.
Then for each $u\in V$, we solve 3SUM for
$A=B=\{w(v): v\in N_G(u)\}\subseteq A_0=B_0$ and $S=W-A-w(u)$.
(We can exclude using a number twice or thrice in solving 3SUM by
a standard trick of appending each number with two or three extra bits.)
The $n$ instances of 3SUM can be solved in $\OO(n^{1.9})$
time each, after preprocessing in $\OO(n^2)$ expected time.
(The result can be made deterministic, as we can
afford to switch to the slower
$O((1/\alpha)^{0.4651}n^{2.3729})$-time preprocessing algorithm.)
\end{proof}

The above ``application'' is admittedly contrived but demonstrates
the potential usefulness of solving 3SUM in preprocessed universes.
(Vassilevska Williams and Williams~\cite{VW09} had a more general
result for counting
the number of copies of any constant-size subgraph
with a prescribed total vertex weight, but their bound is not
subcubic for 4-vertex subgraphs.)

For another application, we can reduce 3SUM for $(K,L,M)$-clustered
sets to preprocessing a universe of size $O(K)$
and solving $O(L^3)$ 3SUM instances.  This provides another
explanation why subquadratic algorithms are possible for certain
parameters of clusterability, although the time bounds obtained by
this indirect approach would not be as good as those from Section~\ref{sec:clustered}.

\section{Proof and Time Complexity of the BSG Theorem/Corollary}\label{sec:Details}

\newcommand{\Deg}{\textrm{deg}}
\newcommand{\cdeg}{\textrm{cdeg}}
\newcommand{\BAD}{\textsc{bad}}
\newcommand{\Ex}{\mathbb{E}}

In this section, we review one proof of the
Balog--Szemer\'{e}di--Gowers Theorem, in order to analyze its
construction time and derive the time bound for the BSG Corollary as given by Theorem~\ref{runtime-corollary}, which has been
used in all our applications.
The proof of BSG theorem we will present is due to
Balog~\cite{Balog07} and independently
 Sudakov et al.~\cite{SSV94}, with minor changes
to make it more amenable to algorithmic analysis.  The proof is
based on a combinatorial lemma purely about graphs (Balog and
Sudakov et al.\ gave essentially identical proofs of this
graph lemma, but the latter described a simpler reduction of the
theorem to the lemma).  The entire proof is short (see Sections \ref{sec-graph} and~\ref{sec-bsg-proof}), uses only
elementary arguments, but has intricate details.

To obtain the
best running time, we incorporate a number of
nontrivial additional ideas.  For our randomized time bound,
we need sampling tricks found in sublinear algorithms.
For our deterministic time bound, we need efficient dynamic updates
of matrix products.

\subsection{A Graph Lemma}\label{sec-graph}

\begin{lemma} {\bf (Graph Lemma)}\ \ Given a bipartite graph $G\subseteq A \times B$, with $|G| \geq \alpha |A||B|$, there exist $A'\subseteq A$ and $B'\subseteq B$ such that

\begin{itemize}
\item[\rm (i)]  for every $a' \in A'$, $b' \in B'$, there are $\Omega(\alpha^5 |A||B|)$ length-$3$
       paths from $a'$ to $b'$, and
\item[\rm (ii)] $|G\cap (A'\times B')| \:\geq\: \Omega(\alpha |A'||B|) \:\geq\: \Omega(\alpha^2 |A||B|)$.
\end{itemize}
\end{lemma}

\begin{proof}

Let $N_G(v)$ denote the neighborhood of $v$ in graph $G$.  Let $\Deg_G(v) = |N_G(v)|$.
Let cdeg$_G(u,v) = |N_G(u)\cap N_G(v)|$
(the number of common neighbors of $u$ and $v$, or equivalently,
the number of length-$2$ paths from $u$ to $v$).
The existence of $A'$ and $B'$ is shown via the following
concise but clever algorithm.

\paragraph{Algorithm:}
\

\begin{quote}
\begin{algorithmic}[1]

 \State  $A_0 = \{a \in A: \Deg_G(a) \geq \alpha |B|/2\}$
 \Repeat
 \State pick a random $b^* \in B$
 \State $A^* = A_0\cap N_G(b^*)$
 \State $\BAD^* = \{(a,a') \in A^*\times A^*:$ cdeg$_G(a,a') \leq \alpha^3 |B|/2048\}$
 \Until{$|A^*| \geq \alpha |A|/4$ and $|\BAD^*| \leq \alpha^2 |A^*||A|/256$}
 \State $A' = \{a \in A^*: \Deg_{\BAD^*}(a) \leq \alpha^2 |A|/64\}$
 \State $B' = \{b \in B: \Deg_{G\cap (A'\times B)}(b) \geq \alpha |A'|/4\}$
 \end{algorithmic}
\end{quote}

\paragraph{Correctness of (i):}

Line~6 guarantees that the undirected graph $\BAD^*$ with vertex
set $A^*$ has
at most $\alpha^2 |A^*||A|/256$ edges and thus average degree at most
$\alpha^2 |A|/128$.
From line~7 it follows that $|A'| \geq |A^*|/2 \geq \alpha |A|/8$.

Fix $a' \in A'$ and $b' \in B'$.
By line~8, there are $\geq \alpha|A'|/4 \geq \alpha^2 |A|/32$ vertices $a \in A'$ that are
adjacent to $b'$.  By line~7, all but $\leq \alpha^2 |A|/64$ such vertices $a$ satisfy $(a,a')\not\in \BAD^*$.
By line~5, for each such $a$, there are $\geq  \alpha^3 |B|/2048$ length-$2$ paths from $a'$ to $a$.
Thus, there are $\geq (\alpha^2 |A|/64)\cdot(\alpha^3 |B|/2048) = \Omega(\alpha^5 |A||B|)$
length-$3$ paths from $a'$ to $b'$.

\paragraph{Correctness of (ii):}

Since $A' \subseteq A_0$, by line 1, deg$_G(a') \geq  \alpha |B|/2$ for every $a' \in A'$ and hence $|G\cap (A'\times B)| \geq  |A'|\cdot (\alpha |B|/2)$.
From line $8$  it follows that $|G\cap (A'\times B')| \:\geq\:  |G\cap (A'\times B)| - (\alpha |A'|/4)\cdot |B|
                          \:\geq\:  \alpha |A'||B|/4 \:\geq\:  \alpha^2 |A||B|/32$.

\paragraph{The process ends w.h.p.:}

Line $1$ implies $|G\cap (A_0\times B)| \geq |G| - (\alpha |B|/2)\cdot |A|  \geq  \alpha |A||B|/2$. From line $4$ it follows that
$$\Ex_{b^*}[|A^*|] \:=\: \frac{1}{|B|} \sum_{b^* \in B} |A_0\cap N_G(b^*)| \:=\: \frac{1}{|B|} |G\cap (A_0\times B)| \:\geq\:  \alpha |A|/2.$$
Line $5$ then implies
\begin{eqnarray*}
\Ex_{b^*}[|\BAD^*|] &=& \!\!\!\!\sum_{\substack{a,a' \in A_0:\\\cdeg_G(a,a') \leq  \alpha^3 |B|/2048}}\!\!\!\!\!\!\!\!
                \Pr_{b^*}[a,a' \in N_G(b^*)]
            \ = \!\!\!\!\sum_{\substack{a,a' \in A_0:\\ \cdeg_G(a,a') \leq  \alpha^3 |B|/2048}}\!\!\!\!\!\!\!\!
                \frac{\cdeg_G(a,a')}{|B|}\\
            &\leq&  |A|^2 \cdot \frac{\alpha^3 |B|/2048}{|B|}\ =\ \alpha^3 |A|^2/2048.
\end{eqnarray*}

Define $Z = \alpha^2 |A| (|A^*| - \alpha |A|/4)  - 256\,|\BAD^*|$.
Then $\Ex_{b^*}[Z] \:\geq\:  \alpha^2 |A| (\alpha |A|/4) - \alpha^3 |A|^2/8
\:=\: \alpha^3 |A|^2/8$.
On the other hand, since we always have $Z \leq  \alpha^2 |A|^2$,
$\Ex_{b^*}[Z] \leq  \alpha^2 |A|^2 \Pr_{b^*}[Z > 0]$.
Thus, $\Pr_{b^*}[Z > 0] \:\geq\: (\alpha^3 |A|^2/8)/(\alpha^2 |A|^2) \:=\: \Omega(\alpha)$.

When $Z > 0$, we have simultaneously
$|A^*|\ge\alpha |A|/4$ and $|\BAD^*|\le \alpha^2 |A^*| |A|/256$.
Thus, the number of iterations in lines 2--6 is $\OO(1/\alpha)$ w.h.p.
\end{proof}


\subsection{Time Complexity of the Graph Lemma}\label{sec-graph-time}

\begin{lemma}\label{lem-graph-runtime}
In the Graph Lemma, $A'$ and $B'$ can be constructed by

\begin{itemize}
\item
a deterministic algorithm in $O(\MM(|A|,|A|,|B|))$ time;
\item
a randomized Monte Carlo algorithm in $\OO((1/\alpha)^5 |A'| + (1/\alpha)|B| + (1/\alpha)^6)$ time
(which is correct w.h.p.), given the adjacency matrix of $G$.
\end{itemize}
\end{lemma}

\begin{proof}\

\paragraph{Deterministic time analysis:}
An obvious deterministic implementation would try
all $b^* \in B$ in lines 2--5 until we find one that satisfies
the test in line~6.
For line~4, we can compute $|A^*|$ for all $b^*$ in
$O(|A||B|)$ total time.
For line~5, we can compute $|\BAD^*|$ for all $b^*$ as follows:
First precompute cdeg$_G(a,a')$ for all $a,a' \in A$; this takes
$\MM(|A|,|B|,|A|)$ time by
computing a matrix product $X_1Y_1$ where
$X_1$ is the adjacency matrix of $G$
and $Y_1$ is its transpose.
Let $\BAD_0=\{(a,a')\in A_0\times A_0: \cdeg_G(a,a') \leq  \alpha^3 |B|/2048\}$.
For all $a\in A_0,b^*\in B$, precompute count$(a,b^*) = $
the number of $a'$ with $aa'\in\BAD_0$ and $a'\in N_G(b^*)$;
this takes $\MM(|A_0|,|A_0|,|B|)$ time by
computing a matrix product $X_2Y_2$ where
$X_2$ is the adjacency matrix of $\BAD_0$ and
$Y_2$ is the adjacency matrix of $G\cap (A_0\times B)$.
Then for all $b^*$, we can compute $|\BAD^*|$ by summing
count$(a,b^*)$ over all $a\in N_G(b^*)$.
Lastly, lines 6--7 take $O(|A||B|)$ time.
The total time is $O(\MM(|A|,|B|,|A|))$, since $\MM(\cdot,\cdot,\cdot)$
is known to be invariant under permutation of its three arguments.
This is subcubic in $|A|+|B|$.

\paragraph{Randomized time analysis:}
With Monte Carlo randomization, we now show how to improve the running time significantly to near linear in $|A|+|B|$, which is
\emph{sublinear} in the size of the input adjacency matrix.
To achieve sublinear complexity, we modify the algorithm where $\deg(\cdot)$ and $\cdeg(\cdot)$
are replaced by estimates obtained by \emph{random sampling}.
Let $\delta>0$ be a sufficiently small constant and $N=\sqrt{|A||B|}$.

The following fact will be useful: given a random sample $R\subseteq U$ of size $(1/\delta)^2(1/\alpha)\log N$, for any fixed subset $X$
we can estimate $|X|$ by $|R\cap X|\cdot |U|/|R|$ with additive error $O(\delta\cdot\max\{|X|,\alpha |U|\})$
w.h.p.  This follows from a Chernoff bound.\footnote{
Let $\mu=|X||R|/|U|$.  One version of the Chernoff bound
states that $\Pr[||R\cap X|-\mu| > \delta'\mu] \le
e^{-\Omega(\min\{\delta'^2\mu,\delta'\mu\})}$
(the first term of the min occurs when $\delta'\le 1$, the second
when $\delta'\ge 1$).  Set $\delta'\mu = c\delta\cdot \max\{\mu,\alpha |R|\}$ for an arbitrarily large constant $c$.  Then $\delta'\ge c\delta$ and $\delta'\mu\ge c\delta\alpha|R|$, implying that $\min\{\delta'^2\mu, \delta'\mu\} \ge \Omega(\min\{c^2\delta^2\alpha|R|, c\delta\alpha|R| \})\ge\Omega(c\log N)$.  Thus,
$||R\cap X| - \mu| \le O(\delta\cdot \max\{\mu,
\alpha|R|\})$ w.h.p.  Finally, multiply both sides by $|U|/|R|$.
}
In particular, w.h.p., $|R\cap X|\ge \alpha |R|$ implies $|X|\ge (1-O(\delta))\alpha |U|$, and $|R\cap X|\le \alpha |R|$ implies
$|X|\le (1+O(\delta))\alpha |U|$.

In line $1$, we draw a random sample $R_1\subseteq B$ of size $(1/\delta)^2(1/\alpha) \log N$.
Then for each $a \in A$, we can replace deg$_G(a)$ by
$|\{b\in R_1: (a,b)\in G\}| \cdot |B|/|R_1|$
with additive error $O(\delta\cdot \max\{\deg_G(a), \alpha|B|\})$ w.h.p.
This gives $A_0$ in $\OO((1/\alpha) |A|)$ time.

Line 4 takes $O(|A|)$ time.

In line 5, we draw another (independent) random sample $R_5\subseteq B$ of size $(1/\delta)^2(1/\alpha)^3 \log N$.
Then for each $a,a' \in A^*$, we can replace cdeg$_G(a,a')$ by
$|\{b\in R_5: (a,b),(a',b)\in G\}| \cdot |B|/|R_5|$
with additive error $O(\delta\cdot \max\{\cdeg_G(a,a'), \alpha^3|B|\})$ w.h.p.
We do not explicitly construct $\BAD^*$; rather, we can \emph{probe} any entry of
the adjacency matrix of $\BAD^*$ in $\OO((1/\alpha)^3)$ time.

In line~6, we draw another random sample $R_6\subseteq A^*\times A^*$ of size $(1/\delta)^2(1/\alpha)^2 \log N$.
Then we can replace $|\BAD^*|$ by
$|\{(a,a')\in R_6: (a,a')\in \BAD^*\}| \cdot |A^*|^2/|R_6|$
with additive error $O(\delta\cdot \max\{|\BAD^*|, \alpha^2 |A^*|^2\})$ w.h.p.
This takes $\OO((1/\alpha)^2)$ probes to $\BAD^*$, and thus $\OO((1/\alpha)^5)$ time.

Recall that the number of iterations for lines 2--6 is $\OO(1/\alpha)$ w.h.p.
Thus, the total time for lines 2--6 is $\OO((1/\alpha) |A| + (1/\alpha)^6)$.

In line 7, we draw another random sample $R_7\subseteq A^*$ of size $(1/\delta)^2(1/\alpha)^2 \log N$.
Then for each $a \in A^*$, we replace
deg$_{\BAD^*}(a)$ by
$|\{a'\in R_7: (a,a')\in \BAD^*\}| \cdot |A^*|/|R_7|$
with additive error $O(\delta\cdot\max\{\deg_{\BAD^*}(a),\alpha^2 |A^*|\})$ w.h.p.
This takes a total of $\OO((1/\alpha)^2 |A^*|)$
probes to $\BAD^*$, and thus $\OO((1/\alpha)^5 |A^*|)
= \OO((1/\alpha)^5 |A'|)$ time.

In line~8,  we draw one final random sample $R_8\subseteq A'$ of size $(1/\delta)(1/\alpha) \log N$.
Then for each $b \in B$, we can replace deg$_{G\cap (A'\times B)}(b)$ by
$|\{a\in R_8: (a,b)\in G\}| \cdot |A'|/|R_8|$
with additive error $O(\delta\cdot\max\{\deg_{G\cap (A'\times B)}(b), \alpha |A'|\})$ w.h.p.
This takes $\OO((1/\alpha)|B|)$ time.


The overall running time is $\OO((1/\alpha)|A|  + (1/\alpha)^5 |A'| + (1/\alpha)|B| + (1/\alpha)^6)$.
Since $|A'| \geq  \Omega(\alpha |A|)$, the first term can
be dropped.
The correctness proofs of (i) and (ii) still go through
after adjusting all constant factors by ${} \pm O(\delta)$, if we make
$\delta$ small enough.
\end{proof}

(We could slightly improve the $\alpha$-dependencies in the randomized time bound by incorporating matrix multiplication,
but they are small enough already that such improvements will not affect the final cost in our applications.)

\subsection{Proof of the BSG Theorem}\label{sec-bsg-proof}

We claim that the subsets $A'$ and $B'$ from the Graph Lemma
already satisfy the conditions stated in the BSG Theorem.
It suffices to verify condition (i).
To this end, let $S=\{a+b : (a,b)\in G\}$ and imagine
the following process:

\begin{quote}
\begin{algorithmic}

 \For {each $c \in A'+B'$}
 \State   take the lexicographically smallest $(a',b') \in  A'\times B'$ with $c = a'+b'$
 \For {each length-$3$ path $a'bab' \in G$}
 \State mark the triple $(a'+b, a+b, a+b')\in S^3$
\EndFor
\EndFor
 \end{algorithmic}
\end{quote}

By (i) in the Graph Lemma, the number of marks is at least
$\Omega(|A'+B'|\cdot \alpha^5 |A||B|)$.
On the other hand, observe that each triple $(a'+b, a+b, a+b')\in S^3$ is
marked at most once, because from the triple,
$c=(a'+b) - (a+b) + (a+b')$ is determined, from which
$a'$ and $b'$ are determined, and from which
$a=(a+b')-b'$ and $b=(a'+b)-a'$ are determined.
Thus, the number of marks is at most $|S|^3$.

Putting the two together, we get
$|A'+B'| \:\le\: O((1/\alpha)^5 |S|^3/(|A||B|))
\:=\: O((1/\alpha)^5 t^3 N)$.
\qed

\bigskip
The running time for the BSG Theorem is thus as given in Lemma~\ref{lem-graph-runtime}.

\subsection{Proof of the BSG Corollary}\label{sec-bsg-corollary-proof}

Note that although the BSG Corollary statement
has $|A||B|=O(N^2)$, we may assume that $|A||B|=\Theta(N^2)$, since
we can change parameters to $\hat{N}=\sqrt{|A||B|} = O(N)$,
$\hat{t}=tN/\hat{N}=\Omega(t)$, and $\hat{\alpha}=\alpha (N/\hat{N})^2=\Omega(\alpha)$.
Then $\hat{\alpha}\hat{N}^2 = \alpha N^2$, and
$(1/\hat{\alpha})^5\hat{t}^3\hat{N} =
 O((1/\alpha)^5 t^3 N)$.

We can construct the subsets $A_1,\ldots,A_k,B_1,\ldots,B_k$ and
the remainder set $R$ in the BSG Corollary, simply by
repeated applications of the BSG Theorem:

\begin{quote}
\begin{algorithmic}[1]
  \State $G_1 = \{(a,b)\in A \times B: a+b\in S\}$
  \For {$i=1,2,\ldots$}
    \State if $|G_i| \leq  \alpha N^2$ then set $k=i-1$, $R=G_{k+1}$, and return
    \State apply the BSG Theorem to $G_i$
with parameter $\alpha_i = \frac{|G_i|}{N^2}$
       to get subsets $A_i,B_i$
    \State  $G_{i+1} = G_i \setminus (A_i\times B_i)$
  \EndFor
 \end{algorithmic}
\end{quote}
\IGNORE{
In line~3, the parameters in the BSG Theorem are
$\hat{N}=\sqrt{|A||B|}\le N$, $\hat{t}=|S|/\hat{N}$, and $\hat{\alpha}_i=|G_i|/\hat{N}^2\ge \alpha$.
Note that $(1/\hat{\alpha}_i)^5\hat{t}^3\hat{N} \:=\:
(\hat{N}^2/|G_i|)^5 (|S|/\hat{N})^3 \hat{N}
\:\le\: (N^2/|G_i|)^5 (|S|/N)^3 N
\:\le\: (1/\alpha)^5 t^3 N$.
}

A naive upper bound on $k$ would be $O((1/\alpha)^2)$, since $\Omega(\alpha^2 N^2)$ edges are removed in each iteration.
For a more careful analysis, observe that
$$|G_{i+1}| \:\leq\:  |G_i| - \Omega(\alpha_i^2 N^2)
           \:=\: |G_i| \cdot \left(1 - \Omega\left(\frac{|G_i|}{N^2}\right)\right),$$
which implies that
$$ \frac{N^2}{|G_{i+1|}} \:\geq\:  \frac{N^2}{|G_i|} \cdot \left(1 + \Omega\left( \frac{|G_i|}{N^2}\right)\right)
                 \:=\:  \frac{N^2}{|G_i|} + \Omega(1).$$
Iterating, we get $N^2/|G_k| \geq  \Omega(k)$.
Since $|G_k| \geq  \alpha N^2$, we conclude that $k \leq  O(1/\alpha)$.
\qed


\subsection{Time Complexity of the BSG Corollary}\label{sec-bsg-corollary-time}

We now analyze the running time for the BSG Corollary.
We may assume that $|A|\le |B|$ without loss of generality.
We may also assume that $t\ge \alpha N/|A|$,
because otherwise
$|\{(a,b):a+b\in S\}|\:\le\: |S||A| \:\le\: \alpha N^2$ and so the trivial solution with
$k=0$ would work.
We may further assume that $N\ge (1/\alpha)^5 t^3$,
because otherwise $(1/\alpha)^5 t^3 N \ge N^2$, and so the trivial
solution with $k=1,A_1=A,B_1=B$ would already work.
Putting the two assumptions together, we have
$N\:\ge\: (1/\alpha)^5 (\alpha N/|A|)^3 \:=\: (1/\alpha)^2 (N/|A|)^3\ge (1/\alpha)^2 N/|A|$, and
so $|A|\ge (1/\alpha)^2$.

The following additional fact will be useful:
$\sum_{i=1}^k A_i = \OO(|A|)$.
To see this, observe that
$$ |G_{i+1}| \:\leq\:  |G_i| - \Omega(\alpha_i |A_i||B|)
             \: =\:  |G_i| \cdot \left(1 - \Omega\left(\frac{|A_i|}{|A|}\right)\right),$$
which implies that
$$|G_k| \:\leq\:  |G_1| \cdot e^{-\Omega( \sum_{i=1}^{k-1} |A_i|/|A| )}\ \ \Longrightarrow\ \
\sum_{i=1}^k |A_i| \:\leq\:  O\left(|A|\log\frac{|G_1|}{|G_k|}\right)
                    \:\leq\:  O(|A| \log(1/\alpha)).$$

This fact implies that the total cost of updating the
adjacency matrix as edges are deleted in line~5 is
at most $O(\sum_{i=1}^k |A_i||B_i|) \le O(\sum_{i=1}^k |A_i||B|) = \OO(N^2)$.
Furthermore, we can construct all the
sumsets $A_i+B_i$ naively, again in total time
$\OO(\sum_{i=1}^k |A_i||B_i|) = \OO(N^2)$.
It thus remains to bound the total cost of the invocations
to the BSG Theorem in line~4.

\paragraph{Deterministic time analysis.}
For the deterministic version of the algorithm,
we can naively upper-bound the total time of all
$k=O(1/\alpha)$ iterations by $O((1/\alpha)\MM(|A|,|A|,|B|))$.
We show how to improve the $\alpha$-dependency slightly.
To achieve the speedup, we modify the implementation
of the deterministic algorithm in
the Graph Lemma to support \emph{dynamic updates}
in $G$, namely, deletions of subsets of edges.

Suppose we delete $A_i\times B_i$ from $G$.
All the steps in the algorithm can be redone in
at most $O(|A||B|)$ time,
except for the computation of the products $X_1Y_1$
and $X_2Y_2$.  As $A_i\times B_i$ is deleted, $X_1$
undergoes changes to $|A_i|$
rows of $X_1$.  We can compute the
change in $X_1Y_1$ by multiplying the change in $X_1$
(an $|A_i|\times|B|$ matrix if the all-zero rows are ignored),
with the matrix $Y_1$, in $\MM(|A_i|,|B|,|A|)$ time.
Now, $Y_1$ also undergoes changes to $|A_i|$ columns.
We can similarly update $X_1Y_1$ under these changes
in $\MM(|A|,|B|,|A_i|)$ time.

The product $X_1Y_1$ itself
undergoes changes in $|A_i|$ rows and columns, and so does
the next matrix $X_2$.  Moreover, $X_2$ undergoes $z_i$
additional row and column deletions where $z_i$ is the number of
deletions to $A_0$.  Also, $Y_2$ undergoes
$|A_i|$ row changes and $z_i$ row deletions.
We can update $X_2Y_2$
under changes to $|A_i|+z_i$ rows in $X_2$ in
$\MM(|A_i|+z_i,|A|,|B|)$ time.  Next we can update $X_2Y_2$
under changes to $|A_i|+z_i$ columns in $X_2$ in
$\MM(|A|,|A_i|+z_i,|B|)$ time.   Lastly we can update $X_2Y_2$ under
changes to $|A_i|+z_i$ rows in $Y_2$ in $\MM(|A|,|A_i|+z_i,|B|)$ time.

Recall that $\MM(\cdot,\cdot,\cdot)$ is invariant under
permutation of its arguments, and $\sum_{i=1}^k |A_i| = \OO(|A|)$.  Moreover, $\sum_{i=1}^k z_i\le |A|$,
since $A_0$ undergoes only deletions.
The overall running time is
thus $O(\sum_{i=1}^k \MM(|A_i|+z_i,|A|,|B|))
\:\le\: O\left(\sum_{i=1}^k \left\lceil \frac{|A_i|+z_i}{\alpha |A|}\right\rceil\cdot \MM(\alpha |A|,|A|,|B|)\right) \:=\: O((1/\alpha) \MM(\alpha|A|,|A|,|B|))$.

According to known upper bounds on rectangular matrix multiplication \cite{HuangPan,LeGall1,LeGall2,Vas},
$$\MM(\alpha |A|, |A|,|A|)\:=\: O(\alpha^{\frac{2.3729-2}{1-0.3029}} |A|^{2.3729})\:=\:O(\alpha^{0.5349}|A|^{2.3729})$$ for
$\alpha |A| \gg |A|^{0.3029}$, which is true
since $|A|\ge (1/\alpha)^2$ by assumption.
So our time bound is
$O((1/\alpha) \MM(\alpha|A|,|A|,|B|))
\:\le\:O((1/\alpha) (|B|/|A|)\cdot \MM(\alpha|A|,|A|,|A|))
\:=\:O((1/\alpha)^{0.4651}|A|^{1.3729}|B|) \:=\:
O((1/\alpha)^{0.4651} N^{2.3729})$.

\paragraph{Randomized time analysis.}

For the randomized version of the algorithm, we can bound
the total time by $$\OO\left(N^2 + \sum_{i=1}^k ((1/\alpha)^5 |A_i| + (1/\alpha)|B| + (1/\alpha)^6)\right)\:=\: \OO(N^2 + (1/\alpha)^5 |A| + (1/\alpha)^2|B| + (1/\alpha)^7).$$
The third and fourth terms can be dropped, because they are
dominated by the first and second since $|A|\ge (1/\alpha)^2$ by assumption.
In the case $t\ge 1$, the second term can also be dropped,
because it is dominated by the first
since $N\ge (1/\alpha)^5 t^3$ by assumption.

Since we can efficiently check whether the solution is correct,
the Monte Carlo algorithm can
be converted into a Las Vegas algorithm.
This completes the proof of Theorem~\ref{runtime-corollary}.
\qed


\bigskip
A gap remains between the deterministic and randomized results.
For constant $\alpha$, we believe it should be possible
to reduce the deterministic running time in the BSG Corollary
to $\OO(N^2)$, by replacing
matrix multiplication with FFT computations, but we are currently
unable to bound the $\alpha$-dependency polynomially in the time
bound (roughly
because as we
iterate, the graph $G_i$ becomes less and less well-structured).

\section{Proof of the FFT Lemma}\label{sec-fft}

To complete the last piece of the puzzle, we now supply a proof of the FFT Lemma.
Note that although the statement of the FFT Lemma has $A,B\in [U]^d$,
we may assume that $d=1$, since we can map
each point $(x_1,\ldots,x_d)$ to an integer
$\sum_{i=1}^d x_i(2U)^{i-1}$.

The idea is to use hash to a smaller universe and
then solve the problem on the smaller universe by FFT\@.
As our problem involves sumsets,
we need a hash function that is ``basically'' additive.
The following definition suffices for our purposes: we say
that a function $h$ is \emph{pseudo-additive} if there
is an associated function $\hat{h}$ such that
$\hat{h}(h(a)+h(b)) = h(a+b)$ for every $a,b$.
For example, the function $h_p(x)= x\bmod{p}$ is pseudo-additive
(with the associated function $\hat{h}_p=h_p$).

\newcommand{\collide}{\textrm{collide}}
\newcommand{\HH}{{\cal H}}

We do not need a single perfect hash function (which would be
more time-consuming to generate and may not be pseudo-additive); rather,
it suffices to have a small number
of hash functions that ensure each element in $T$ has no collisions
with respect to at least one hash function.  To be precise,
define $\collide(h,x)=\{y\in T\setminus\{x\}: h(y)=h(x)\}$.
We say that
a family $\HH$ of functions is \emph{pseudo-perfect} for $T$
if for every $x\in T$ there is an $h\in \HH$ with $|\collide(h,x)|=0$.
\begin{description}
\item[New Problem:] Construct a family $\HH$ of $k$ pseudo-additive functions from $[U]$ to $[\OO(N)]$ that is pseudo-perfect for $T$.
\end{description}

\paragraph{Computing $A+B$, given a pseudo-perfect pseudo-additive family for $T$.}

\newcommand{\hh}{\hat{h}}
Given such an $\HH$, we can compute $A+B$ as follows.
For each $h\in\HH$, we first compute $h(A)+h(B)$ by FFT in
$\OO(N)$ time and obtain $\hh(h(A)+h(B))$.
Then for each $s\in T$, we identify an
$h\in\HH$ with $|\collide(h,s)|=0$, and if $h(s)\in \hh(h(A)+h(B))$, we report $s$.
The total time of the whole algorithm is $\OO(kN)$
(assuming that each $h$ and $\hh$ can be evaluated in constant time).
To prove correctness, just note that
for $a\in A,b\in B$, we have
$h(s)=\hh(h(a)+h(b))$ iff $h(s)=h(a+b)$ iff $s=a+b$,
assuming that $|\collide(h,s)|=0$ (since $a+b\in A+B\subseteq T$).

It remains to solve the problem of constructing the hash functions $\HH$.

\paragraph{A standard randomized construction of
a pseudo-perfect pseudo-additive family.}
With randomization, we can simply pick $k=\log N + 1$ random primes $p\in [cN\log^2 U]$ for a sufficiently large constant $c$,
and put each function $h_p(x)=x\bmod{p}$ in $\HH$.

To see why this works, consider a fixed $x\in T$.
For any $y\in X\setminus\{x\}$, the number of primes $p$ with $y\bmod{p}=x\bmod{p}$ is equal to the number of prime divisors of
$|x-y|$
and is at most $\log U$.  Since the number of primes in $[cN\log^2 U]$
is at least $2N\log U$ for a sufficiently large $c$ by the
prime number theorem, $\Pr_p[y\bmod{p}=x\bmod{p}]\le 1/(2N)$.
Thus, $\Pr_p[ |\collide(h_p,x)|\neq 0] \le 1/2$.
So, $\Pr[\forall h_p\in\HH, |\collide(h_p,x)|\neq 0]\le 1/2^k \le 1/(2N)$.  Therefore, the overall failure probability is at most $1/2$.

Note that we can compute the numbers
$|\collide(h,x)|$ for all $x\in T$ for any given hash function
in linear time after assigning elements to buckets.
In particular, we can verify whether the
construction is correct in $\OO(N)$ time.
We conclude that there is a Las Vegas algorithm with total expected time  $\OO(N)$.

\paragraph{A new deterministic construction of a pseudo-perfect pseudo-additive family.}
An obvious way to derandomize the previous method is to try
all primes in $[cN\log^2 U]$ by brute force, but the running time
would be at least
$\Omega(N^2)$.  Indeed, that was the approach taken by
Amir et al.~\cite{AKP07}.  We describe a faster deterministic
solution by replacing a large prime with multiple smaller
primes, using hash functions of the form
$h_{p_1,\ldots,p_\ell}(x)=(x\bmod{p_1},\ldots, x\bmod{p_\ell}) \in\Z^\ell$.  Such a function remains pseudo-additive
(with the associated function $\hat{h}_{p_1,\ldots,p_\ell}(x_1,\ldots,x_\ell)
= (x_1\bmod{p_1},\ldots, x_\ell\bmod{p_\ell})$).
The idea is to generate the $\ell$ smaller primes in $\ell$
separate rounds.  The algorithm works as follows:

\begin{quote}
\begin{algorithmic}[1]

 \State  $S=T$
 \While{$|S|\neq\emptyset$}
   \For{$i=1$ to $\ell$}
     \State pick a prime $p_i\in [cN^{1/\ell}\log^2 U]$ with
     \State \ \ \ \
$|\{x\in S: |\collide(h_{p_1,\ldots,p_i},x)| < N^{1-i/\ell}\}|
\:\ge\: |S|/2^i$
   \EndFor
   \State put $h_{p_1,\ldots,p_\ell}$ in $\HH$, and
remove all $x$ with $|\collide(h_{p_1,\ldots,p_\ell},x)|=0$ from $S$
 \EndWhile
 \end{algorithmic}
\end{quote}

Consider the inner for loop.  Lines 4--5 take $\OO(N^{1+1/\ell})$
time by brute force.  But why does $p_i$ always exist?
Suppose that $p_1,\ldots,p_{i-1}$ have already been chosen, and imagine that $p_i$ is picked at random.
Let $C_i(x)$ be a shorthand for $\collide(h_{p_1,\ldots,p_i},x)$.
Consider a fixed $x\in S$ with $|C_{i-1}(x)| < N^{1-(i-1)/\ell}$.
For any fixed $y\in T\setminus\{x\}$, $\Pr_{p_i}[y\bmod{p_i}=x\bmod{p_i}]\le 1/(2N^{1/\ell})$ by
an argument we have seen earlier.
Thus, $$\Ex_{p_i}[|C_i(x)|] \:=\: \Ex_{p_i}[|\{y\in C_{i-1}(x): y\bmod{p_i}=x\bmod{p_i}\}|]
\:\le\: |C_{i-1}(x)|/(2N^{1/\ell}) < N^{1-i/\ell}/2.$$
By Markov's inequality,
$\Pr_{p_i}[|C_i(x)| < N^{1-i/\ell}] \ge 1/2$.
Since we know from the previous iteration that there
are at least $|S|/2^{i-1}$ elements $x\in S$
with $|C_{i-1}(x)| < N^{1-(i-1)/\ell}$, we then have
$\Ex_{p_i}[|\{x\in S: |C_i(x)< N^{1-i/\ell}\}|]\ge |S|/2^i$.
So there exists $p_i$ with the stated property.

Line~7 then removes at least $|S|/2^\ell$ elements.
Hence, the
number of iterations in the outer while loop is
$k\le \log N/\log(\frac{2^\ell}{2^\ell-1})=O(2^\ell\log N)$.
Each function $h_{p_1,\ldots,p_\ell}$
maps to $[\OO(N^{1/\ell})]^\ell$,
which can easily be mapped back to one dimension in $[\OO(N)]$
while preserving pseudo-additivity,
for any constant $\ell$.  We conclude that there
is a deterministic algorithm with total running time
$\OO(N^{1+1/\ell})$ for an arbitrarily large constant $\ell$.
This gives $\OO(N^{1+\eps})$.  (More precisely, we can bound
the total deterministic time by $\OO(N 2^{O(\sqrt{\log N\log\log U})})$ by choosing a nonconstant $\ell$.)
%
\qed

\begin{remark}
In the above results, the $\OO$ notation hides
not only $\log N$ but also $\log U$ factors.
For many of our applications, $U=N^{O(1)}$ and so this is
not an issue.  Furthermore, in the randomized version,
we can use a different hash function~\cite{BDP08} to reduce $U$ to
$N^{O(1)}$ first, before running the above algorithm.
In the deterministic version, it seems possible to
lower the dependency on $U$ by using recursion.
\end{remark}

\begin{remark}
Our hash function family construction has other applications,
for example, to the \emph{sparse convolution}
problem: given two nonnegative vectors $u$ and $v$, compute
their classical convolution $\vec{u}*\vec{v} = \vec{z}$ (where $z_k =
\sum_{i=0}^k u_iv_{k-i}$) in ``output-sensitive'' time,
close to $||\vec{z}||$, the number of nonzeros in $\vec{z}$.
The problem was raised by Muthukrishnan~\cite{Muthukrishnan95},
and previously solved by Cole and Hariharan~\cite{CH02}
with a randomized Las Vegas algorithm in $\OO(||\vec{z}||)$ time.

Let $A=\{a: u_a\neq 0\}$ and $B=\{b: v_b\neq 0\}$.  Then
$||\vec{z}||$ is precisely $|A+B|$.
If we are given a superset $T$ of $A+B$ of size $O(||\vec{z}||)$, we can solve the problem deterministically
using a pseudo-perfect pseudo-additive family $\HH$ as follows:
For each $h\in\HH$,
precompute the vectors $\vec{u'}$ and $\vec{v'}$ of length $\OO(||\vec{z}||)$ where
$u'_i=\sum_{a:h(a)=i}u_a$ and $v'_j=\sum_{y:h(b)=j}v_b$.
Compute $\vec{u'}*\vec{v'}=\vec{z'}$ in $\OO(||\vec{z}||)$ time by FFT\@.
Compute $z''_\ell = \sum_{k:\hh(k)=\ell}z'_k$.
Then for each $s\in T$ with $|\collide(h,s)|=0$,
set $z_s=z''_{h(s)}$.
To prove correctness, first observe that
$z'_k=\sum_{a,b:h(a)+h(b)=k}u_av_b$.
If $|\collide(h,s)|=0$, then
$z''_{h(s)} = \sum_{a,b:\hh(h(a)+h(b))=h(s)}u_av_b
=\sum_{a,b: h(a+b)=h(s)} u_av_b = \sum_{a,b:a+b=s} u_av_b=z_s$.

Cole and Hariharan's algorithm is more general and does not require the superset $T$
to be given.  It would be interesting to obtain a deterministic
algorithm that similarly avoids $T$.
\end{remark}

\begin{remark}
Another application is
\emph{sparse wildcard matching}.
The problem is: given a pattern and text that are sparse with few non-zeroes, find all alignments where every non-zero pattern element matches the text character aligned with it. Sparse wildcard matching has applications, such as subset matching, tree pattern matching, and geometric pattern matching; see~\cite{CH02}.

Cardoze and Schulman~\cite{CS98} proposed a Monte Carlo near-linear-time algorithm. Cole and Hariharan~\cite{CH02} transformed this into a Las Vegas algorithm. Amir et al.~\cite{AKP07} considered the indexing version of this problem where the text is preprocessed for subsequent pattern queries. In this setting they proposed an $\OO(N^2)$ preprocessing time algorithm, where $N$ is the number of non-zeroes in the text. The query time is then $\OO(N)$. The latter is essentially based upon the construction of a deterministic pseudo-perfect
pseudo-additive family for $T=\{i:\ \mbox{$i$ is a non-zero text location}\}$.
It can be checked that our new deterministic solution is
applicable here and thus improves their preprocessing time to $\OO(N^{1+\eps})$, yielding the first quasi-linear-time deterministic
algorithm for sparse wildcard matching.
%
\end{remark}

\section{Final Remarks}

We have given the first truly subquadratic algorithms for
a variety of problems related to 3SUM\@.
Although there is potential for improving the exponents in
all our results, the main contribution is that we have broken the
barrier.

An obvious direction for improvement would be to reduce the
$\alpha$-dependency in the BSG Theorem itself;
our work provides more urgency towards this well-studied combinatorial
problem.  Recently, Schoen~\cite{Sch14} has announced such an improvement of
the BSG Theorem, but it is unclear whether this result will be useful
for our applications for two reasons: First, the time complexity of
this construction needs to be worked out. Second, and more importantly,
Schoen's improvement is for a more basic version of the theorem without $G$,
and the extension with $G$ is essential to us.

Here is one specific mathematical question that is particularly relevant to us:
\begin{quote}
Given subsets $A,B,S$ of
an abelian group of size $N$, we want to
cover $\{(a,b)\in A\times B: a+b\in S\}$
by bicliques $A_i\times B_i$, so as to minimize the cost function
$\sum_i |A_i+B_i|$.  (There is no constraint on the
number of bicliques.)  Prove worst-case bounds on the minimum
cost achievable as a function of $N$.
\end{quote}
A bound $O(N^{13/7})$ follows from the BSG Corollary, simply by
creating $\alpha N^2$ extra ``singleton'' bicliques
to cover $R$, and setting $\alpha$ to minimize
$O(\alpha N^2 + (1/\alpha)^6 N)$.  An improvement on this
combinatorial bound would have implications to at
least one of our algorithmic applications, notably
Theorem~\ref{thm-preproc1}.

We hope that our work will inspire further applications of additive combinatorics in algorithms.
For instance, we have yet to study special cases of $k$SUM for
larger $k$;
perhaps some multi-term extension of the BSG Theorem~\cite{BC} would be
useful there.  As an extension of bounded monotone (min,+)
convolution, we may also consider (min,+) matrix
multiplication for the case of integers in $[n]$ where the
rows and columns satisfy monotonicity or the bounded differences
property.
It would be exciting if the general integer 3SUM or APSP
problem could be solved using tools from
additive combinatorics.

}

\bibliographystyle{plain}
\bibliography{bsg}

\end{document}